\documentclass[11pt]{article}
\usepackage{fullpage}
\usepackage[onehalfspacing]{setspace}
\usepackage{xcolor}
\usepackage{amsmath}
\usepackage{amsthm}
\usepackage{graphicx,amssymb}
\usepackage{dsfont}
\usepackage{mathtools}
\usepackage{newpxtext}
\usepackage{newpxmath}
\usepackage{algorithm}
\usepackage[noend]{algpseudocode}
\usepackage[most]{tcolorbox}
\usepackage{multicol}
\usepackage{caption}
\usepackage{parskip}
\usepackage{tikz}
\usetikzlibrary{positioning, arrows.meta}
\usepackage{booktabs,tabularx,threeparttable}
\usepackage{subcaption}
\usepackage{hyperref}
\usepackage[nameinlink,capitalize]{cleveref}

\title{Distortion of Metric Voting with Bounded Randomness}

\date{}

\newtheorem{theorem}{Theorem}[section]
\newtheorem{lemma}[theorem]{Lemma}

\newtheorem{corollary}[theorem]{Corollary}

\newtheorem{fact}[theorem]{Fact}

\theoremstyle{definition}
\newtheorem{definition}[theorem]{Definition}

\newtheorem{remark}[theorem]{Remark}

\newtheorem{claim}[theorem]{Claim}
\newtheorem*{theorem*}{Theorem}

\AddToHook{env/lemma/begin}{\crefalias{theorem}{lemma}}
\AddToHook{env/conjecture/begin}{\crefalias{theorem}{conjecture}}
\AddToHook{env/corollary/begin}{\crefalias{theorem}{corollary}}
\AddToHook{env/proposition/begin}{\crefalias{theorem}{proposition}}
\AddToHook{env/fact/begin}{\crefalias{theorem}{fact}}
\AddToHook{env/definition/begin}{\crefalias{theorem}{definition}}
\AddToHook{env/example/begin}{\crefalias{theorem}{example}}
\AddToHook{env/remark/begin}{\crefalias{theorem}{remark}}
\AddToHook{env/question/begin}{\crefalias{theorem}{question}}
\AddToHook{env/condition/begin}{\crefalias{theorem}{condition}}
\AddToHook{env/claim/begin}{\crefalias{theorem}{claim}}

\AddToHook{cmd/appendix/before}{\crefalias{section}{appendix}}

\crefname{theorem}{Theorem}{Theorems}
\crefname{lemma}{Lemma}{Lemmas}
\crefname{conjecture}{Conjecture}{Conjectures}
\crefname{corollary}{Corollary}{Corollaries}
\crefname{proposition}{Proposition}{Propositions}
\crefname{fact}{Fact}{Facts}
\crefname{definition}{Definition}{Definitions}
\crefname{example}{Example}{Examples}
\crefname{remark}{Remark}{Remarks}
\crefname{question}{Question}{Questions}
\crefname{condition}{Condition}{Conditions}

\crefname{Program}{Program}{Programs}
\creflabelformat{Program}{(#2\textup{#1})#3}

\definecolor{boxc}{rgb}{0.5, 0.2, 0.6}
\definecolor{linkc}{rgb}{0.6, 0.2, 0.3}
\definecolor{citec}{rgb}{0.3, 0.2, 0.6}
\definecolor{urlc}{rgb}{0.2, 0.6, 0.3}
\hypersetup{
    colorlinks=true,
    linkcolor=linkc,
    citecolor=citec,
    urlcolor=urlc
}

\newtcolorbox[auto counter,number within=section]{mybox}[2][]{colback=boxc!8!white,colframe=boxc!60!black,title={#2},#1}

\newcommand{\R}{\ensuremath\mathbb{R}}
\newcommand{\N}{\ensuremath\mathbb{N}}

\newcommand{\defeq}{\coloneq}
\newcommand{\set}[1]{\ensuremath{\left\{#1\right\}}}
\newcommand{\abs}[1]{\ensuremath{\left\vert#1\right\vert}}

\newcommand{\veps}{\ensuremath\varepsilon}
\newcommand*\diff{\mathop{}\!\mathrm{d}}
\renewcommand*\d{\mathop{}\!\mathrm{d}}

\newcommand{\Unif}{\mathtt{Unif}}

\DeclareMathOperator*{\Supp}{\mathrm{Supp}}
\DeclareMathOperator*{\argmin}{arg\,min}

\DeclareMathOperator*{\Ex}{\mathbb{E}}
\newcommand{\E}[2]{\ensuremath{\Ex\limits_{#1}{\left[#2\right]}}}
\DeclareMathOperator*{\Prob}{\mathbf{Pr}}
\renewcommand{\Pr}[2]{\ensuremath{\Prob\limits_{#1}{\left[#2\right]}}}

\newcommand{\SC}{\mathrm{SC}}
\newcommand{\ML}{\mathrm{ML}}
\newcommand{\SL}{\mathrm{SL}}
\newcommand{\SDL}{\mathrm{SDL}}
\newcommand{\eML}[1]{{#1\text{-}\ML}}
\newcommand{\kSL}[1]{{#1\text{-}\SL}}
\newcommand{\kSDL}[1]{{#1\text{-}\SDL}}
\newcommand{\ekSL}[2]{{\left(#1, #2\right)\text{-}\SL}}

\begin{document}
\begin{titlepage}

\author{\begin{tabular}{cc}
\begin{minipage}[t]{0.45\textwidth}\centering
Ziyi Cai\\
\small Rutgers University\\
\small\href{mailto:zc417@cs.rutgers.edu}{zc417@cs.rutgers.edu}
\end{minipage}
&
\begin{minipage}[t]{0.45\textwidth}\centering
D. D. Gao\\
\small Rutgers University\\
\small\href{mailto:d.gao@rutgers.edu}{d.gao@rutgers.edu}
\end{minipage}
\\[8ex]
\begin{minipage}[t]{0.45\textwidth}\centering
Prasanna Ramakrishnan\\
\small Stanford University\\
\small\href{mailto:pras1712@stanford.edu}{pras1712@stanford.edu}
\end{minipage}
&
\begin{minipage}[t]{0.45\textwidth}\centering
Kangning Wang\\
\small Rutgers University\\
\small\href{mailto:kn.w@rutgers.edu}{kn.w@rutgers.edu}
\end{minipage}
\end{tabular}}

\maketitle
\thispagestyle{empty}
\setcounter{page}{0}

\begin{abstract}
We study the design of voting rules in the metric distortion framework. It is known that any deterministic rule suffers distortion of at least $3$, and that randomized rules can achieve distortion strictly less than $3$, often at the cost of reduced transparency and interpretability. In this work, we explore the trade-off between these paradigms by asking whether it is possible to break the distortion barrier of $3$ using only ``bounded'' randomness. We answer in the affirmative by presenting a voting rule that (1) achieves distortion of at most $3 - \varepsilon$ for some absolute constant $\varepsilon > 0$, and (2) selects a winner uniformly at random from a deterministically identified list of constant size. Our analysis builds on new structural results for the distortion and approximation of Maximal Lotteries and Stable Lotteries.
\end{abstract}
 
\newpage

\tableofcontents
\thispagestyle{empty}
\setcounter{page}{0}

\end{titlepage}

\section{Introduction}
Elections are fundamental to democratic governance, making the design of voting rules a matter of critical importance. One particularly significant format is the single-winner election, in which we select a single winner from a pool of candidates.

Most classical voting rules are deterministic: they can be viewed as functions that map voter preferences to a unique winning candidate. However, randomized voting rules might have a richer history and theoretical presence than one might assume. In Athenian democracy, sortition was employed to select representatives via random sampling \cite{headlam1891election}. In modern social choice theory, randomization is central to prominent voting rules such as Random Dictatorship \cite{gibbard1977manipulation,amar1984choosing}, which selects the favorite candidate of a randomly chosen voter, and Maximal Lotteries \cite{MISC:conf/kreweras1965aggregation,MISC:journals/fishburn1984probabilistic,brandt2017rolling}, which selects a winner according to a mixed-strategy Nash equilibrium. Randomization also serves as a natural method for tie-breaking and is closely related to fractional power allocation. For example, in ancient Rome, two consuls were elected every year and would rotate leadership roles every month \cite{abbott1901history}; such a fractional allocation, in many mathematical models, is equivalent to a randomization between the two consuls. For a comprehensive discussion on randomized voting rules, see the survey by Brandt \cite{brandt2017rolling}.

Theoretical research often places no limits on how randomized voting rules utilize randomness; instead, they are simply defined as functions that map voter preferences to probability distributions over candidates. While mathematically elegant, this view often conflicts with practical considerations. For instance, it would not have been feasible for the Romans to elect a vast number of consuls associated with a complicated weight vector, with each consul leading for a duration proportional to his weight. Voters can also feel uneasy with randomized voting rules, which may give a positive probability to every candidate, even if some are viewed as broadly unacceptable. These conflicts motivate us to place restrictions on the form of randomness. In particular, we ask: can we follow the Roman example and elect candidates to fill a constant number of ``consulships'' with equal power?

Such a method is highly interpretable and transparent: one only needs to publish a constant-size roster (possibly with duplicate outcomes), from which the outcome is drawn uniformly at random. This restrictive form of randomness has already been explored in several social choice settings. Flanigan, Kehne, and Procaccia \cite{DBLP:conf/nips/FlaniganKP21} studied choosing from a collection of ``fair'' panels in the context of citizens' assemblies. Ebadian, Filos-Ratsikas, Latifian, and Shah \cite{DBLP:conf/nips/EbadianFL023} analyzed the distortion of voting rules that choose uniformly from a committee of distinct alternatives in the normalized social choice framework \cite{DBLP:conf/cia/ProcacciaR06}.

Similar to \cite{DBLP:conf/nips/EbadianFL023}, our work focuses on voting rules with limited randomness; unlike their normalized distortion setting, we adopt the framework of metric distortion \cite{DBLP:journals/ai/AnshelevichBEPS18,DBLP:journals/jair/AnshelevichP17}. In the model, candidates and voters are assumed to reside in a common metric space; the voting rule, having access only to the ordinal preferences of the voters, aims to minimize the social cost---the sum of distances from all voters to the winning candidate. A voting rule is said to have a metric distortion of at most $c \in \mathbb{R}^+$ if, for every metric space, the expected social cost of the winner is at most $c$ times the optimal social cost. It has been established that deterministic voting rules can achieve a distortion of $3$ \cite{DBLP:conf/focs/GkatzelisHS20,DBLP:conf/ijcai/KizilkayaK22,DBLP:conf/sigecom/Kizilkaya023} but not better \cite{DBLP:journals/ai/AnshelevichBEPS18}; on the other hand, randomized rules can achieve a distortion of $3 - \varepsilon$ for an absolute constant $\varepsilon > 0$ \cite{DBLP:journals/jacm/CharikarRWW24,DBLP:conf/sigecom/CharikarRTW25}.

In this sense, $3$ is the critical threshold that separates randomized voting rules from deterministic ones. We ask whether we can break this threshold using voting rules with only ``bounded randomness,'' and specifically, we pose the following question:
\begin{quote}
Are there absolute constants $\veps > 0$ and $k \in \mathbb{N}$ such that there exists a voting rule with metric distortion at most $3 - \varepsilon$ that uniformly randomizes among at most $k$ options?
\end{quote}

\subsection{Our Results}

We answer the question above in the affirmative. En route, we establish several structural results regarding the distortion and approximation of Maximal Lotteries and their generalization, Stable Lotteries \cite{DBLP:journals/teco/ChengJMW20,DBLP:conf/sigecom/CharikarRTW25}.

We adopt the analytical framework of biased metrics \cite{DBLP:conf/soda/CharikarR22,DBLP:journals/jacm/CharikarRWW24}. This framework characterizes the ``hardest'' metric spaces: as shown in the work of \cite{DBLP:conf/soda/CharikarR22,DBLP:journals/jacm/CharikarRWW24}, if a voting rule achieves a distortion of at most $c$ for all biased metrics, then it also achieves a distortion of at most $c$ for all metric spaces. Consequently, restricting our attention to these metrics incurs no loss of generality. Furthermore, as observed by \cite{DBLP:journals/jacm/CharikarRWW24}, the notion of $(\alpha, \beta)$-consistency (\cref{def:consistent}) is central to the analysis; indeed, all known proofs of voting rules achieving a distortion constant less than $3$ rely on this framework, which we will follow as well.

We now describe the key ingredients of our proof, each of which may be of independent interest in the study of these social choice rules.

\paragraph{Distortion robustness of Maximal Lotteries (\cref{sec:ml_support}).}
Maximal Lotteries have particularly appealing performance in terms of metric distortion: their expected metric distortion is at most $3$, and it is strictly better than $3$ when the biased metric is not sufficiently consistent \cite{DBLP:journals/jacm/CharikarRWW24}. However, our analysis requires a robust guarantee that holds beyond expectation. What we show for this purpose is that selecting \emph{any} candidate in the support of Maximal Lotteries yields a metric distortion of at most $4 + \sqrt{17} < 8.124$. A direct implication of this result is that the rules proposed by \cite{DBLP:journals/jacm/CharikarRWW24}---which currently hold the best-known distortion guarantees---are also robust: they maintain bounded distortion regardless of the realization of the random choice.

\paragraph{Distortion of Approximate Maximal Lotteries (\cref{sec:ml_distortion}).}
Since we impose the constraint that our voting rule must only randomize over a constant number of options, Maximal Lotteries, which can have a linear-sized support in general, cannot be directly applied as part of our voting rule. We take inspiration from a recent work by Charikar, Ramakrishnan, and Wang \cite{DBLP:conf/soda/CharikarRW26}, who proved that the distribution formed by a constant number of samples from Maximal Lotteries has a positive probability of being a mixed strategy in an approximate Nash equilibrium. In our work, we extend this message to the metric distortion setting and prove the following result: the distribution formed by a constant number (which depends on a parameter $\varepsilon$) of samples from Maximal Lotteries has a positive probability to guarantee a metric distortion of at most $3 + \varepsilon$, for any prespecified constant $\varepsilon > 0$. It might be a surprising fact that this guarantee does not depend on the number of voters and candidates.

Although our result conveys a similar message to that of \cite{DBLP:conf/soda/CharikarRW26}, the proof of our result is not merely a simple reduction from theirs. Instead, we have to delve into the biased metric framework to carefully bound several components of the approximation loss.

\paragraph{Approximate Stable Lotteries (\cref{sec:stable_lotteries}).}
When the underlying biased metric is sufficiently consistent, Stable Lotteries---a generalization of Maximal Lotteries---offer strong distortion guarantees \cite{DBLP:conf/sigecom/CharikarRTW25}. Again, given our constraints on the form of randomness, we must prove that an approximate version of it can have similar guarantees. We demonstrate here that---analogous to the result of \cite{DBLP:conf/soda/CharikarRW26} for Maximal Lotteries---a uniform distribution over a constant number of samples from Stable Lotteries can, with a positive probability, ``uniformly'' (i.e., without dependence on the number of voters and candidates) approximate the guarantees of exact Stable Lotteries.

\paragraph{Putting things together: breaking the distortion barrier of $3$ with bounded randomness.}
With these components, we answer our main question affirmatively: a uniform distribution over a constant number of options can achieve a metric distortion constant strictly less than $3$.

\begin{theorem*}[Main Theorem, Informal]
    There exists a randomized voting rule that, with appropriately chosen constant probabilities, runs either
    \begin{itemize}
        \item an approximate Maximal Lottery;
        \item an approximate Stable Lottery on a carefully selected subset of candidates.
    \end{itemize}
    This rule achieves metric distortion strictly less than $3$.
\end{theorem*}

To obtain a deterministic polynomial-time algorithm to find such a distribution, we can enumerate all multisets (from small to large) of the candidates and stop if the induced distribution achieves a metric distortion of less than $3 - \varepsilon$. Its running time is guaranteed to be polynomial since, according to our main result, there exists a multiset of constant size with an induced metric distortion of less than $3 - \varepsilon$.

Our result has a notable implication for the committee selection setting described by \cite{DBLP:journals/corr/abs-2507-17063}, where the goal is to select $k$ winning candidates to minimize the social cost (defined as the double sum of distances across all voter--winner pairs). It is known that a deterministic committee selection rule can achieve a metric distortion of $3$ \cite{DBLP:journals/corr/abs-2507-17063}, and our result implies that if each candidate is allowed to occupy multiple seats, then for any sufficiently large $k$, there exists a \emph{deterministic} committee selection rule with a distortion constant strictly less than $3$. This connection hints at an underlying connection between committee selection and our model of limited randomness.

\subsection{Further Related Work}

The metric distortion framework has inspired a rich body of work over the last decade. We refer the reader to \cite{DBLP:conf/ijcai/AnshelevichF0V21} for a detailed survey.

The model was first proposed by \cite{DBLP:conf/aaai/AnshelevichBP15,DBLP:journals/ai/AnshelevichBEPS18}, who evaluated the metric distortion of several classical deterministic voting rules. Their analysis left the optimal distortion achievable by a deterministic voting rule as an intriguing open problem, with an upper bound of $5$ proven for the Copeland rule, and a lower bound of $3$ conjectured to be optimal. The upper bound was later improved to $2 + \sqrt{5}$ by \cite{DBLP:conf/ec/MunagalaW19} and finally to $3$ by \cite{DBLP:conf/focs/GkatzelisHS20}, both of which introduced completely novel voting rules in the process. \cite{DBLP:conf/ijcai/KizilkayaK22,DBLP:conf/sigecom/Kizilkaya023} later proposed simpler novel voting rules achieving the optimal distortion of $3$.

In parallel, \cite{DBLP:conf/sigecom/FeldmanFG16,DBLP:journals/jair/AnshelevichP17} showed that the simple Random Dictatorship rule achieves distortion $3$, matching the best possible for deterministic rules. Despite a lower bound of $2$ that was conjectured to be optimal by \cite{DBLP:conf/sigecom/GoelKM17}, improvements were elusive. This conjecture was eventually disproven by \cite{DBLP:conf/soda/CharikarR22}, who proved a lower bound of $2.112$, and \cite{pulyassary2021randomized}, who independently provided a lower bound of $2.063$. Recently, \cite{DBLP:journals/jacm/CharikarRWW24} gave the first constant improvement to the upper bound, introducing a new voting rule with distortion of at most $2.753$. Closing the gap between $2.112$ and $2.753$ remains a central open problem.

Besides the metric distortion framework's usefulness in evaluating existing voting rules and motivating new ones, it has also been fruitful for understanding the power of voting rules in alternative preference models, both under information constraints and when additional information is available. In the constrained setting, \cite{DBLP:conf/sigecom/GoelKM17,FL25,DBLP:conf/sigecom/CharikarRTW25} study kinds of tournament rules that solely use the aggregate comparisons between pairs of candidates. \cite{DBLP:conf/aaai/Kempe20b} considered voting rules with limited communication (also a motivating factor in \cite{DBLP:conf/ijcai/KizilkayaK22}), while \cite{DBLP:conf/aaai/GrossAX17} study voting rules where voters only rank their top few candidates. Considering additional information,
\cite{DBLP:conf/sigecom/BergerFGT24} study metric distortion in the learning-augmented setting, and \cite{DBLP:conf/stoc/GoelGM25,DBLP:journals/corr/abs-2511-00986} show that deliberations between voters can improve distortion guarantees.

\section{Preliminaries and Notation}
\subsection{Social Choice and Voting}
We consider a standard voting setting. Let $V$ be a set of $n$ voters, and $C$ be a set of $m$ candidates. A \emph{preference profile} over $V$, denoted as $\succ_V$, is a function that maps each voter $v$ to a strict total order $\succ_v$ of the candidates $C$. That is, we write $i \succ_v j$ if voter $v$ prefers candidate $i$ to candidate $j$. An \emph{election instance} is then given by a tuple $\mathcal{E} = (V, C, \succ_V)$. A \emph{deterministic voting rule} maps an election instance $\mathcal{E}$ to a single candidate in $C$, while a \emph{randomized voting rule} maps $\mathcal{E}$ to a probability distribution over $C$.

Let $\mathcal{P}(v)$ be a predicate on voters. We define $S_\mathcal{P}$ to be the set of voters that satisfy $\mathcal{P}$ and let $s_\mathcal{P} \defeq \frac1n \abs{S_\mathcal{P}}$. For example, for candidate $i, j \in C$, we write $s_{i \succ j}$ to denote the fraction of voters who prefer candidate $i$ to candidate $j$. To simplify subsequent definitions and proofs, we stipulate that $s_{j \succ j} = \frac12$ for each candidate $j \in C$.

We sometimes interpret the quantities $s_{i \succ j}$ for candidates $i, j \in C$ through the lens of a weighted tournament graph. Specifically, this is a directed graph whose vertex set is $C$, and which contains, for every ordered pair of distinct candidates $i, j \in C$, an edge from $i$ to $j$ with weight $s_{i \succ j}$. 

For a set of candidates $I \subseteq C$ and a candidate $j \in C$, we let $s_{I \succ j}$ denote the fraction of voters who prefer every candidate in $I$ to $j$ (and vice versa for $s_{j \succ I}$). If $j \in I$, then $S_{I \succ j} = \varnothing$ and $s_{I \succ j} = 0$. 

We also allow $\mathcal{P}(v)$ to be a randomized indicator. For instance, given a candidate $i$ and a distribution $D$ over the candidates, we define
\[
s_{i \succ D} \defeq \frac1n\sum_{v \in V}\E{j \sim D}{\mathds{1}\left[i \succ_v j\right]} = \frac1n\sum_{v \in V}\Pr{j \sim D}{i \succ_v j}.
\]
More generally, for two distributions $D$ and $D^\prime$ over the candidates, we define
\[
s_{D \succ D^\prime} \defeq \frac1n\sum_{v \in V}\E{i \sim D, j \sim D^\prime}{\mathds{1}\left[i \succ_v j\right]} = \frac1n\sum_{v \in V}\Pr{i \sim D, j \sim D^\prime}{i \succ_v j}.
\]

The following simple fact relates the quantities $s_{i \succ j}, s_{I \succ j}$, and $s_{D \succ j}$.
\begin{fact}\label{fac:set2dis}
  For any candidate set $I \subsetneq C$ and any distribution $D$ supported on $I$, if $j \not\in I$, then
  \[
  s_{I \succ j} \leq \min_{i \in I}s_{i \succ j} \leq s_{D \succ j} \leq \max_{i \in I}s_{i \succ j}  ~~~ \text{and} ~~~ s_{j \succ I} \leq \min_{i \in I}s_{j \succ i} \leq s_{j \succ D} \leq \max_{i \in I}s_{j \succ i}.
  \]
\end{fact}
\begin{proof}[Proof of \cref{fac:set2dis}]
  By definition \[S_{I \succ j} = \set{v \in V:I \succ_v j} = \bigcap_{i \in I} \set{v \in V: i \succ_v j},\] Hence $S_{I \succ j} \subseteq S_{i \succ j}$ for all candidate $i \in I$, which implies $s_{I \succ j} \leq \min_{i \in I}s_{i \succ j}$. Since $s_{D \succ j}$ is the expectation of $s_{i \succ j}$ over $i \sim D$, it must lie between $\min_{i \in I}s_{i \succ j}$ and $\max_{i \in I}s_{i \succ j}$.

  The second chain of inequalities follows by an analogous argument, exchanging the positions of $i$ and $j$.
\end{proof}

Throughout this work we write $D(I)$ to denote the distribution $D$ conditioned on the selected candidate lying in $I$.

We use the term \emph{strategy} to refer to either a deterministically selected candidate $j$ or a probability distribution $D$ over the candidates $C$. The following lemma establishes a triangle inequality over strategies.

\begin{lemma}[see, e.g.\@, {\cite[Claim 2]{DBLP:journals/jacm/CharikarRWW24}}]\label{lem:triangle}
  For any three strategies $X, Y$, and $Z$, we have \[s_{X \succ Y} \leq s_{X \succ Z} + s_{Z \succ Y}.\]
\end{lemma}
\subsection{Metric Distortion}
In this work, we adopt the standard \emph{metric distortion} framework. In this setting, the set of voters $V$ and the set of candidates $C$ are embedded in a common pseudometric space $(V \sqcup C, d)$.

\begin{definition}[Pseudometric Space] A \emph{pseudometric space} is an ordered pair $(M, d)$ where $M$ is a set and the function $d \colon M \times M \rightarrow \R_{\geq 0}$ is a \emph{pseudometric} on $M$. They satisfy the following axioms for all $x, y, z \in M$:
  \begin{enumerate}
    \item Identity: $d(x, x) = 0$.
    \item Symmetry: $d(x, y) = d(y, x)$.
    \item Triangle inequality: $d(x, z) \leq d(x, y) + d(y, z)$.
  \end{enumerate}
\end{definition}

Given a pseudometric $d$, the \emph{social cost} of selecting a candidate $i$ is defined as \[\SC(i) \defeq \frac1n \sum_{v \in V}d(i, v).\] We denote by $i^*$ a candidate of minimum social cost, i.e.\@, \[i^* \in \argmin_{i \in C}\SC(i).\]

We now define \emph{metric distortion}.
\begin{definition}[Metric Distortion]
Given a (possibly randomized) voting rule $f$, the \emph{metric distortion} of $f$ is the smallest ratio $\alpha \geq 1$ such that for all pseudometric spaces, the (expected) social cost of the candidate selected by $f$ is at most $\alpha$ times the social cost of selecting any candidate.
\end{definition}

It is worth noting that the voting rule $f$ does not have access to the underlying pseudometric $d$. Instead, it observes only the preference profile $\succ_V$ which is assumed to be consistent with $d$. That is, for any candidate $i, j \in C$ and any voter $v \in V$, we have $i \succ_v j$ holds only if $d(i, v) \leq d(j, v)$.

\subsection{Biased Metrics}
\emph{Biased metrics} were introduced in \cite{DBLP:conf/soda/CharikarR22} and have since been shown to be a powerful tool for analyzing metric distortion \cite{DBLP:conf/soda/CharikarR22,DBLP:journals/jacm/CharikarRWW24,DBLP:conf/sigecom/CharikarRTW25}. Our work also builds on this technique. We review this framework here for completeness.

\begin{definition}[Biased Metrics] Let $(x_1, \dots, x_m) \in \R^m_{\geq 0}$ be a nonnegative vector such that $x_{i^*} = 0$ for some $i^* \in [m]$. The \emph{biased metric} induced by $(x_1, \dots, x_m)$ is defined as follows. For every voter $v$ and every candidate $j$, let
  \begin{align*}
  d(i^*, v) &\defeq \frac12\max_{i, j: i\succeq_v j}(x_i - x_j),\\
  d(j, v) - d(i^*, v) &\defeq \min_{k:j\succeq_v k}x_k.
  \end{align*}
\end{definition}

Note that under a biased metric, we always have $d(j, v) \geq d(i^*, v)$ for all voters $v$ and candidate $j$. Intuitively, a biased metric assumes $d(j, i^*) = x_j$ for each candidate $j \in C$, and places voters as close as possible to candidate $i^*$, while pushing voters as far as possible from candidates other than $i^*$, subject to consistency.

The work of \cite{DBLP:conf/soda/CharikarR22} shows that it suffices to consider biased metrics when computing metric distortion. Specifically, for any (possibly randomized) voting rule $f$ and any metric $d^\prime$ that satisfies $d^\prime(j, i^*) = x_j$ for every candidate $j \in C$, the social cost of $f$ under the biased metric $d$ induced by $(x_1, \dots, x_m)$ is at least its social cost under $d^\prime$.

We now show how to derive metric distortion bounds using these biased metrics. Fix a vector $(x_1, \dots, x_m) \in \R^m_{\geq 0}$ with $x_{i^*} = 0$, and let $d$ denote the biased metric induced by $(x_1, \dots, x_m)$. For $t\geq 0$, define $I_t \defeq \set{k \in C:x_k \leq t}$. We begin with two basic facts.

\begin{fact}\label{fac:d_to_s_1} For any $t \geq 0$, we have
  \[\Pr{v \sim V}{d(j, v) - d(i^*, v) > t} = s_{I_t \succ j}.\]
\end{fact}
\begin{fact}\label{fac:d_to_s_2} For any $t \geq 0$, we have
  \[\Pr{v \sim V}{2d(i^*, v) \leq t} = s_{\forall i \succ j, x_i - x_j \leq t}.\]
\end{fact}
\begin{proof}[Proof of \cref{fac:d_to_s_1}]
  By definition, for all voter $v$ we have
  \[
  d(j, v) - d(i^*, v) > t \iff \min_{k: j\succeq_v k}x_k > t \iff x_k > t \text{ if } j \succeq_v k \iff k \succ_v j \text{ if } x_k \leq t.
  \] Therefore, voter $v$ prefers every candidate in $I_t$ to $j$, or equivalently, we have $I_t \succ_v j$.
\end{proof}
\begin{proof}[Proof of \cref{fac:d_to_s_2}] By definition, for all voter $v$ we have
  \[
  2d(i^*, v) \leq t \iff \max_{i, j:i \succeq_v j}(x_i - x_j) \leq t \iff x_i - x_j \leq t \text{ if } i \succeq_v j.
  \] Since $x_i - x_j = 0 \leq t$ always holds when $i = j$, the condition is equivalent to $\forall i \succ_v j, x_i - x_j \leq t$.
\end{proof}

The metric distortion of a voting rule $f$ can be characterized by the following theorem.

\begin{theorem}[\cite{DBLP:journals/jacm/CharikarRWW24}]\label{thm:biased_metric} Let $f$ be a voting rule that selects candidate $j$ with probability $p_j$ and let $D$ denote the induced distribution over candidates. The metric distortion of $f$ is at most $1 + 2\lambda$, if and only if, for every preference profile and every vector $(x_1, \dots, x_m) \in \R_{\geq 0}^m$ that satisfies $x_{i^*}=0$, the following condition holds:
  \[
  \int_0^\infty\ell(D, t) \diff t\leq \lambda \int_0^\infty r(t) \diff t,
  \] where
  \[
  \ell(D, t) = \sum_{j \notin I_t}s_{I_t\succ j}p_j \quad \text{ and } \quad r(t) = 1 - s_{\forall i \succ j, x_i - x_j \leq t}.
  \]
\end{theorem}

We will write $L(D) = \int_0^\infty\ell(D, t) \diff t$ and $R = \int_0^\infty r(t) \diff t$ throughout the remainder of the paper. \cref{fig:biased_metric} illustrates the monotone functions $\ell(D, t)$ and $r(t)$; the area corresponding to $L(D)$ is shaded in blue, while the one corresponding to $R$ is shaded in red.

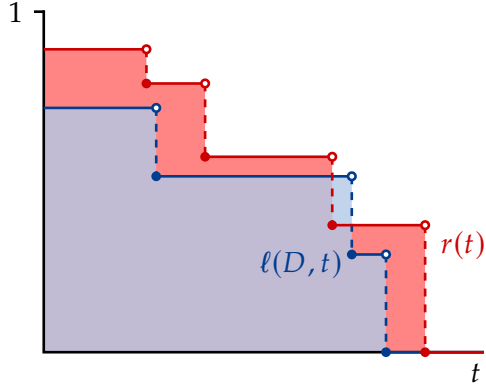
\begin{figure}[t!]
    \centering
        \begin{tikzpicture}[
    scale=0.65,
myredline/.style={color=red!80!black, line width=1pt},
    myredfill/.style={color=red!60, opacity=0.8},
    myblueline/.style={color=blue!40!teal!80!black, line width=1pt},
    mybluefill/.style={color=blue!40!teal!30, opacity=0.8},
    axis/.style={line width=1pt, black},
    dashed guide/.style={dashed, thin, black!80, dash pattern=on 5pt off 3pt},
dot/.style={circle, fill=white, draw=#1, line width=1pt, inner sep=1pt},
filldot/.style={circle, fill=#1, inner sep=1.3pt}
]

\def\yTop{7}
    \def\yB{4.2}
    \def\yHalf{3.5}
    \def\yZero{0}
    \def\xMax{9}
    \def\xTau{4.0}
    \def\xAlphaR{5.0}
    \def\xMu{5.9}

    \begin{scope}
\fill[myredfill] 
            (0,0) -- (0,6.2) -- (2.1,6.2) -- (2.1,5.5) -- (3.3,5.5) -- 
            (3.3,4.0) -- (5.9, 4.0) -- 
            (5.9,2.6) -- (7.8,2.6) -- (7.8,0) -- (\xMax,0) -- cycle;
    \end{scope}

    \fill[mybluefill] 
        (0,\yZero) -- (0,5) -- (2.3,5) -- (2.3,3.6) -- (6.3,3.6) -- (6.3,2) -- (7,2) -- (7,\yZero) -- cycle;

\draw[axis] (0, \yTop) -- (0, 0) -- (\xMax, 0) node[below, xshift=-0.1cm] {$t$};

    \draw[line width=1pt] (0, \yTop - 0.03) -- (-0.2, \yTop - 0.03) node[left] {$1$};

\coordinate (B_start) at (0, 5);
    \coordinate (B1_R) at (2.3, 5); \coordinate (B1_low) at (2.3, 3.6);
    \coordinate (B2_R) at (6.3, 3.6); \coordinate (B2_low) at (6.3, 2);
    \coordinate (B3_R) at (7, 2); \coordinate (B3_low) at (7, 0);
    \coordinate (B_end) at (\xMax, 0);

    \draw[myblueline] (B_start) -- (B1_R);
    \draw[myblueline] (B1_low) -- (B2_R);
    \draw[myblueline] (B2_low) -- (B3_R);
    \draw[myblueline] (B3_low) -- (B_end);

    \draw[myblueline, dashed] (B1_R) -- (B1_low);
    \draw[myblueline, dashed] (B2_R) -- (B2_low);
    \draw[myblueline, dashed] (B3_R) -- (B3_low);

\def\blueColor{blue!40!teal!80!black}
    
    \node[dot=\blueColor] at (B1_R) {};
    \node[filldot=\blueColor] at (B1_low) {};
    
    \node[dot=\blueColor] at (B2_R) {};
    \node[filldot=\blueColor] at (B2_low) {};
    
    \node[dot=\blueColor] at (B3_R) {};
    \node[filldot=\blueColor] at (B3_low) {};

\coordinate (R_start) at (0, 6.2);
    \coordinate (R1_R) at (2.1, 6.2); \coordinate (R1_low) at (2.1, 5.5);
    \coordinate (R2_R) at (3.3, 5.5); \coordinate (R2_low) at (3.3, 4.0);
    \coordinate (R3_R) at (5.9, 4.0); \coordinate (R3_low) at (5.9, 2.6);
    \coordinate (R4_R) at (7.8, 2.6); \coordinate (R4_low) at (7.8, 0);
    \coordinate (R_end) at (\xMax, 0);

    \draw[myredline] (0, 6.2) -- (R1_R);
    \draw[myredline] (R1_low) -- (R2_R);
    \draw[myredline] (R2_low) -- (R3_R);
    \draw[myredline] (R3_low) -- (R4_R);
    \draw[myredline] (R4_low) -- (R_end);

    \draw[myredline, dashed] (R1_R) -- (R1_low);
    \draw[myredline, dashed] (R2_R) -- (R2_low);
    \draw[myredline, dashed] (R3_R) -- (R3_low);
    \draw[myredline, dashed] (R4_R) -- (R4_low);

\def\redColor{red!80!black}

    \node[dot=\redColor] at (R1_R) {};
    \node[filldot=\redColor] at (R1_low) {};
    \node[dot=\redColor] at (R2_R) {};
    \node[filldot=\redColor] at (R2_low) {};
    \node[dot=\redColor] at (R3_R) {};
    \node[filldot=\redColor] at (R3_low) {};
    \node[dot=\redColor] at (R4_R) {};
    \node[filldot=\redColor] at (R4_low) {};

\node[blue!40!teal!80!black, anchor=west] at (4.2, 1.8) {$\ell(D, t)$};
    
\node[red!80!black, anchor=south west] at (7.9, 1.7) {$r(t)$};

\end{tikzpicture}

         \caption{$\ell(D, t)$ and $r(t)$.}
        \label{fig:biased_metric}
\end{figure}

\begin{proof}[Proof of \cref{thm:biased_metric}]
Fix a vector $(x_1, \dots, x_m)$ and let $d$ be the biased metric induced by $(x_1, \dots, x_m)$.

On the one hand, for any candidate $j$, we can compute
\begin{align*}
\SC(j) - \SC(i^*) &= \E{v \sim V}{d(j, v) - d(i^*, v)}\\
&= \int_0^\infty \Pr{v\sim V}{d(j, v) - d(i^*, v) > t}\diff t \tag{$d(j, v) - d(i^*, v) \geq 0$}\\
&= \int_0^\infty s_{I_t \succ j}\diff t. \tag{\cref{fac:d_to_s_1}}
\end{align*}
When $j$ is sampled according to the distribution $D$, it follows that
\[
\sum_{j \in C}p_j\left(\SC(j) - \SC(i^*)\right) = \int_0^\infty\sum_{j \in C}s_{I_t \succ j}p_j \diff t = \int_0^\infty\ell(D, t)\diff t.
\]

On the other hand, we have
\begin{align*}
2\SC(i^*) &= \E{v\sim V}{2d(i^*, v)}\\
&= \int_0^\infty \Pr{v\sim V}{2d(i^*, v) > t}\diff t\\
&= \int_0^\infty \bigl(1 - s_{\forall i \succ j, x_i - x_j \leq t}\bigr) \diff t. \tag{\cref{fac:d_to_s_2}}
\end{align*}
The proof is concluded by observing that
\[
\frac{\sum_{j \in C}p_j\SC(j)}{\SC(i^*)} = 1 + 2 \cdot \frac{\sum_{j \in C}p_j\left(\SC(j) - \SC(i^*)\right)}{\SC(i^*)}
\]
and applying the given condition.
\end{proof}

Sometimes, it is more convenient to work with the following sufficient condition.

\begin{corollary}[\cite{DBLP:conf/soda/CharikarR22,DBLP:journals/jacm/CharikarRWW24}]\label{cor:biased_metric}
The metric distortion of $f$ is at most $1 + 2\lambda$ if for all preference profiles, the following condition holds:
\[
\sum_{j\in J}s_{i^* \succ j}p_j \leq \lambda(1 - s_{i^* \succ J}) \qquad\text{for all }J\subseteq C \setminus \set{i^*}.
\]
\end{corollary}
\begin{proof}[Proof of \cref{cor:biased_metric}]
Fix a preference profile and a biased metric consistent with it. A sufficient condition for \cref{thm:biased_metric} is that $\ell(D, t) \leq \lambda \cdot r(t)$ for all $t \geq 0$. To verify this, it suffices to show
\[
s_{I_t \succ j} \leq s_{i^* \succ j} \qquad \text{ and } \qquad s_{\forall i \succ j, x_i - x_j \leq t} \leq s_{i^* \succ I_t^c}
\]
in which case setting $J = I_t^c$ yields the claim.

By \cref{fac:set2dis}, since $i^* \in I_t$, we have $s_{I_t \succ j} \leq s_{i^* \succ j}$ for all candidate $j$.

Now suppose that a voter $v$ belongs to $S_{\forall i \succ j, x_i - x_j \leq t}$. For any candidate $k \in I_t^c$, we have $x_k > t$. Hence, if $k \succ_v i^*$ were to hold, then $x_k - x_{i^*} > t$, contradicting the defining condition of the set. Therefore, voter $v$ must prefer $i^*$ to every candidate in $I_t^c$.

This shows that $S_{\forall i \succ j, x_i - x_j \leq t} \subseteq S_{i^*\succ I_t^c}$, and thus $s_{\forall i \succ j, x_i - x_j \leq t} \leq s_{i^* \succ I_t^c}$. The corollary follows.
\end{proof}

\subsection{Maximal Lotteries and Stable \texorpdfstring{$k$}{k}-Lotteries}
We describe two voting rules that play a central role in our approach.
\subsubsection{Maximal Lotteries}
Maximal Lotteries are defined via a constant-sum game between two players, Alice and Bob, known as the \emph{Condorcet game}.

\begin{definition}[Condorcet Game]
  Fix an election instance. Alice picks a distribution $D_A$ and Bob picks a distribution $D_B$ over the candidates simultaneously. Alice and Bob's payoffs are $s_{D_A \succ D_B}$ and $s_{D_B \succ D_A}$, respectively.
\end{definition}

Equivalently, consider the constant-sum game in which Alice and Bob choose candidates (as pure strategies), and Alice's payoff when candidates $a$ and $b$ are chosen is $s_{a \succ b}$. This game admits a mixed-strategy Nash equilibrium, which can be computed in polynomial time. Such a mixed equilibrium corresponds exactly to a Nash equilibrium of the Condorcet game above, since a mixed strategy over candidates is precisely a distribution over candidates.

This observation allows us to define the Maximal Lotteries voting rule.

\begin{definition}[Maximal Lotteries]
  Given a fixed election, compute any symmetric pure-strategy Nash equilibrium $D_\ML$ of the Condorcet game. A \emph{Maximal Lottery} chooses a candidate according to $D_\ML$.
\end{definition}
By a slight abuse of terminology, we also refer to the equilibrium $D_\ML$ itself as the Maximal Lottery.

The value of the Condorcet game is $1/2$, as formalized below.

\begin{theorem}[The Value of the Condorcet game; see, e.g., {\cite{brandt2017rolling}}]
\label{thm:ML_distribution}
  In any election, a Maximal Lottery $D_\ML$ satisfies that for every distribution $D$ over the candidates,
  \[
    s_{D_\ML \succ D} \geq \frac12.
  \]
\end{theorem}

\subsubsection{Stable Lotteries}
Stable $k$-Lotteries \cite{DBLP:journals/teco/ChengJMW20,DBLP:conf/sigecom/CharikarRTW25} are natural generalizations of Maximal Lotteries. Instead of selecting a single candidate from a distribution, they select a multiset of $k$ candidates. For a distribution $D$ over the candidates and an integer $k \in \N^+$, we write $D^k$ for the random multiset obtained by drawing $k$ independent samples from $D$ with replacement.

To define comparisons between multisets, for any multiset $A$, let $n_A(x)$ denote the multiplicity of candidate $x$ in $A$. For any two multisets $A$ and $B$ and any voter $v$, let $x = \max_{\succ_v}\set{A \cup B}$ be the most preferred candidate (according to $v$) appearing in either multiset. We define
\[
\Pr{}{A \succ_v B} \defeq \frac{n_A(x)}{n_A(x) + n_B(x)} \quad \text{ and } \quad \Pr{v \sim V}{A \succ_v B} \defeq \E{v \sim V}{\Pr{}{A \succ_v B}}.
\] These definitions extend naturally to random multisets such as $D^k$.

\begin{definition}[$k$-vs-$1$ Condorcet Game]
  Given a fixed election, Alice picks a distribution $D_A$ and Bob picks a distribution $D_B$ over the candidates simultaneously. Alice and Bob's payoffs are $\Pr{v \sim V}{D^k_A \succ_v D_B}$ and $\Pr{v \sim V}{D_B \succ_v D^k_A}$, respectively.
\end{definition}

In this game, each player's pure strategy is a distribution over candidates, while a mixed strategy is a distribution over such distributions. Thus, a pure-strategy Nash equilibrium is a pair of distributions over candidates. Such an equilibrium always exists in the $k$-vs-$1$ Condorcet game and can be computed in polynomial time \cite[Theorem 5.5]{DBLP:conf/sigecom/CharikarRTW25}. We define the Stable $k$-Lotteries voting rule based on this equilibrium.

\begin{definition}[Stable Lotteries]
  For any $k \in \N^+$, given a fixed election, compute any pure-strategy Nash equilibrium $\bigl(D_\kSL{k},D_\kSDL{k}\bigr)$ of the $k$-vs-$1$ Condorcet game. A \emph{Stable $k$-Lottery} chooses a candidate according to $D_\kSL{k}$ and a \emph{Stable Defensive $k$-Lottery} chooses a candidate according to $D_\kSDL{k}$.
\end{definition}

With a slight abuse of terminology, we also refer to the equilibrium $D_\kSL{k}$ itself as the Stable $k$-Lottery. In particular, a Maximal Lottery is exactly a Stable $1$-Lottery.

The value of the $k$-vs-$1$ Condorcet game is $1 - \frac1{k + 1}$. In particular, we have the following theorem.

\begin{theorem}[The Value of the $k$-vs-$1$ Condorcet Game; {\cite[Definition 5.1]{DBLP:conf/sigecom/CharikarRTW25}}]
In any election and for any $k \in \N^+$, the distribution $D_\kSL{k}$ satisfies that for every distribution $D$ over the candidates,
  \[
    \Pr{v\sim V}{D_\kSL{k}^k \succ D} \geq 1-\frac1{k + 1}.
  \]
\end{theorem}

\section{RepApx Stable Lotteries}
\label{sec:stable_lotteries}
In this section, we first discuss prior work on approximating Maximal Lotteries, and then provide a generalization of the prior result to the case of Stable $k$-Lotteries. 

\subsection{Representative Approximate Lotteries}
An \emph{$\veps$-Nash equilibrium} ($\veps$-NE) \cite{SciDir:journals/jet/Radner80} in a game is an approximation of a Nash equilibrium. For a two-player constant-sum game and an $\veps$-Nash equilibrium in it, each player's expected payoff is at most $\veps$ less than the value of the game.

Since Maximal Lotteries and Stable Lotteries can be characterized as equilibria of two-player constant-sum games, they admit natural $\veps$-approximate counterparts. For a Maximal Lottery (resp.\@ Stable $k$-Lottery), whose value of the game is $\frac12$ (resp.\@ $1 - \frac{1}{k + 1}$), an $\veps$-approximation guarantees that the expected payoff against any opposing strategy is at least $\frac12 - \veps$ (resp.\@ $1 - \frac{1}{k + 1} - \veps$).

We define $\veps$-Representative Approximate Maximal Lotteries and $\veps$-Representative Approximate Stable $k$-Lotteries below. For our definitions, in order to apply these concepts in future sections, we only consider a restricted class of approximate equilibria whose support is contained within that of an exact equilibrium. 

\begin{definition}[$\veps$-Representative Approximate Maximal Lotteries]\label{def:repapx_ML}
  Given an election instance, a distribution $D$ over the candidates is called an \emph{$\veps$-Representative Approximate Maximal Lottery} (abbreviated \emph{$\veps$-RepApx Maximal Lottery}) if both of the following conditions hold:
  \begin{itemize}
    \item there exists a Maximal Lottery $D_\ML$ such that $\Supp(D) \subseteq \Supp(D_\ML)$;
    \item for every candidate $a \in C$, it holds that $s_{a \succ D} \leq \frac12 + \veps$.
  \end{itemize}
\end{definition}

\begin{definition}[$\veps$-Representative Approximate Stable $k$-Lotteries]\label{def:repapx_SL}
  Given an election instance, a distribution $D$ over the candidates is called an \emph{$\veps$-Representative Approximate Stable $k$-Lottery} (abbreviated \emph{$\veps$-RepApx Stable $k$-Lottery}) if both of the following conditions hold:
  \begin{itemize}
    \item there exists a Stable $k$-Lottery $D_\kSL{k}$ such that $\Supp(D) \subseteq \Supp(D_\kSL{k})$;
    \item for every candidate $a \in C$, it holds that $\Pr{v \sim V}{a \succ_v D^k} \leq \frac1{k + 1} + \veps$.
  \end{itemize}
\end{definition}

An $\veps$-RepApx Maximal Lottery with a support size of $O(\varepsilon^{-2})$ always exists, and it can be computed by taking independent samples from a Maximal Lottery; the uniform distribution over these samples constitutes an $\veps$-RepApx Maximal Lottery with high probability \cite{DBLP:conf/soda/CharikarRW26,DBLP:conf/focs/BourneufCT25}.

\begin{mybox}[label={box:emp_ML},nameref={Empirical Maximal Lotteries Sampling}]{$q$-Empirical Maximal Lotteries Sampling}
  \begin{itemize}
    \item Compute any symmetric Nash equilibrium distribution $D_\ML$ of the Condorcet game.
    \item Draw independent samples $c_1, \dots, c_q \sim D_\ML$ with replacement.
    \item Return $\Unif\left(\set{c_1, \dots, c_q}\right)$. (In the returned distribution, each candidate is selected with probability proportional to its multiplicity in the multiset.)
  \end{itemize}
\end{mybox}

\begin{theorem}[Existence of Small-Support $\veps$-RepApx Maximal Lotteries, {\cite[Theorem 1]{DBLP:conf/soda/CharikarRW26}}]\label{thm:repapx_ML_existence}
  For all constant $\veps \in (0, 1)$, in every election, $\left((1+o(1))\frac{\pi}{8}\veps^{-2}\right)$-\nameref{box:emp_ML} outputs an $\veps$-RepApx Maximal Lottery with positive probability, and an $(1+\gamma)\veps$-RepApx Maximal Lottery with probability at least $\frac{\gamma}{1+\gamma}$. 
\end{theorem}

\cref{{thm:repapx_ML_existence}} is not only an existence result but also an algorithmic one: as mentioned in the work of \cite{DBLP:conf/soda/CharikarRW26}, since the sampling process succeeds with constant probability, we can restart the process if the previous one does not succeed, and the expected number of reruns is at most a constant. It is a direct consequence of \cref{{thm:repapx_ML_existence}} that Maximal Lotteries can be $\veps$-approximated by a distribution with support size $O(\veps^{-2})$. This contrasts with support sizes for the general case of $\veps$-Nash equilibrium, where it is known that some games may require a distribution with support size $\Omega(\log(m)\veps^{-2})$, where $m$ in the number of pure strategies \cite{DBLP:conf/sigecom/FederNS07}.

As Stable Lotteries are natural generalizations of Maximal Lotteries, it is natural to ask whether \cref{thm:repapx_ML_existence} also generalizes to the case of Stable $k$-Lotteries. In the remainder of the section, we prove that the more general result does in fact hold via an analogous sampling procedure. 

\begin{mybox}[label={box:emp_SL},nameref={Empirical Stable $k$-Lotteries Sampling}]{$q$-Empirical Stable $k$-Lotteries Sampling}
  \begin{itemize}
    \item Compute any Nash equilibrium distribution $D_\kSL{k}$ of the $k$-vs-$1$ Condorcet game.
    \item Draw independent samples $c_1, \dots, c_q \sim D_\kSL{k}$ with replacement.
    \item Return $\Unif\left(\set{c_1, \dots, c_q}\right)$. (In the returned distribution, each candidate is selected with probability proportional to its multiplicity in the multiset.)
  \end{itemize}
\end{mybox}

\begin{theorem}[Existence of Small-Support $\veps$-RepApx Stable $k$-Lotteries]\label[theorem]{thm:repapx_SL_existence}

For all constants $\veps \in (0, 1)$ and $k\in\N$, in every election, $\left\lceil \frac{\pi}{2}k^2\veps^{-2}\right\rceil$-\nameref{box:emp_SL} outputs an $\veps$-RepApx Stable $k$-Lottery with positive probability, and an $(1+\gamma)\veps$-RepApx Stable $k$-Lottery with probability at least $\frac{\gamma}{1+\gamma}$. 
\end{theorem}

Our proof follows a similar structure to the original proof of \cref{thm:repapx_ML_existence}, in that we will utilize the same known probabilistic machinery to show that there is a positive probability that a uniform distribution over random samples from a Stable $k$-Lottery is a $\veps$-RepApx Stable $k$-Lottery.

\subsection{The DKW Inequality}
In our proof, we will utilize general results about uniform distributions over a multiset of samples from some known distribution. Uniform distributions constructed in this way are known as \emph{empirical distributions}. Given a distribution with cumulative distribution function (CDF) $F$, along with i.i.d.\@ samples $s_1,\ldots,s_q\sim F$, we can formally define the CDF of the empirical distribution as $\widehat{F}(r)\coloneq \frac{1}{q}\sum_{i=1}^q \mathds{1}[s_i\leq r]$. Bounds on the maximum of the difference $\widehat{F}(r)-F(r)$, are well understood and related to statistical tests such as Kolmogorov--Smirnov test for closeness of fit (see, e.g., \cite[Pages 283--287]{dodge2008concise}). Notably, Dvoretzky, Kiefer, and Wolfowitz proved a tight (up to the leading constant) concentration bound for $\sup_r(\widehat{F}(r) - F(r))$ which showed that this difference is unlikely to be too far from $0$  \cite{dvoretzky1956asymptotic}. The original concentration bound of \cite{dvoretzky1956asymptotic} was later improved by Massart to have a tight constant  \cite{massart1990tight}. We state this result formally below. 

\begin{lemma}[\cite{massart1990tight}]\label[lemma]{lem:Massart} For any probability distribution defined by CDF $F$, and its empirical distribution, $\widehat{F}$, constructed from random samples from $F$, we have 
  \[
  \Pr{s_1,\ldots,s_q\sim F}{\sqrt{q} \sup_{r}\left|\widehat{F}(r)-F(r)\right|>\lambda}\leq 2e^{-2\lambda^2}.
  \]
\end{lemma}

The following corollary converts this probability bound to an expectation bound, which we will use in our proof of \cref{thm:repapx_SL_existence}.

\begin{corollary}\label[corollary]{cor:DKW-target}
Letting $d(q)\coloneq \E{s_1,\ldots,s_q\sim [0,1]}{\sup_{r\in [0,1]}\left|\widehat{F}(r)-r\right|}$, we have
\[
d(q)\leq \sqrt{\frac{\pi}{2q}}.
\]
\end{corollary}
\begin{proof}
    When $F$ is a uniform distribution over the interval $[0,1]$, \cref{lem:Massart} gives us that \[\Pr{s_1,\ldots,s_q\sim F}{\sqrt{q} \sup_{r\in [0,1]}\left|\widehat{F}(r)-r\right|>\lambda}\leq 2e^{-2\lambda^2}.\] To convert the probability bound to an expectation bound, note that
    \begin{align*}
        \sqrt{q}\cdot d(q)&\leq \int_0^\infty \Pr{s_1,\ldots,s_q\sim [0,1]}{\sqrt{q} \sup_{r\in [0,1]}\left|\widehat{F}(r)-r\right|>\lambda}\d \lambda \\
        &\leq \int_0^\infty 2e^{-2\lambda^2} \d \lambda = \sqrt{\frac{\pi}{2}}. \qedhere
    \end{align*}
\end{proof}

\subsection{Sampling Process}
In this subsection, we outline a technique that allows us to equate the process of sampling a candidate from a random lottery (i.e. a distribution over a discrete set of candidates) with the process of sampling a point uniformly at random from an interval. This equivalence allows us to apply \cref{cor:DKW-target} to bound the difference between a Stable $k$-Lottery and its empirical distribution. 

Given any random lottery, we can associate each voter $v$ with an interval $[0,1]$ that we can partition in a way that captures both the random lottery distribution and the preference order of $v$. Specifically, given a distribution $D$ that select each candidate $c$ with probability $p_c$, and some voter $v$ with a preference order $c_1^*\prec c_2^*\prec \dots \prec c_m^*$, we partition $[0, 1]$ into consecutive intervals $I_{c_i^*,v}$ with lengths equal to $p_{c_i^*}$. We arrange these intervals in the following way: The candidate $c_1^*$ (who is the least preferred candidate of $v$) is associated with the interval $I_{c_1^*,v}=[0,p_{c_1^*})$. For every $i \in [m-1]$, the upper bound of $I_{c_i^*,v}$ becomes the lower bound of $I_{c_{i + 1}^*,v}$, until we reach $I_{c_n^*,v}=[1-p_{c_n^*},1)$. We use 
the notation
\[
\ell_{c_i^*,v}=\sum_{c : c \prec_v c_i^*}p_c
\]
to denote the lower bound of the interval $I_{c_i^*,v}$. We will associate candidates with points in $[0, 1]$ in this way, and in particular, note that voter $v$ prefers points with larger values (i.e., closer to $1$) over those with smaller ones.

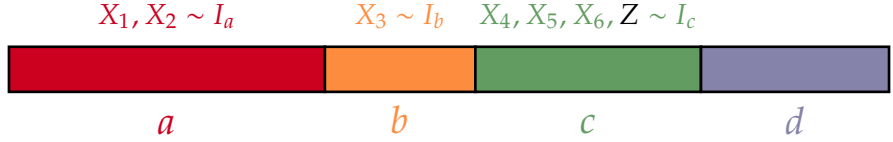
\begin{figure}[t!]
  \centering
  \begin{tikzpicture}[font=\Large\sffamily, thick]

\definecolor{sketchRed}{RGB}{204, 0, 31}    
\definecolor{sketchOrange}{RGB}{254, 140, 62} 
\definecolor{sketchGreen}{RGB}{98, 156, 96} 
\definecolor{sketchBlue}{RGB}{133, 131, 170}  

    \node at (5.5, 2.5) {
        $S = \left\{ \textcolor{sketchRed}{a, a}, \textcolor{sketchOrange}{b}, \textcolor{sketchGreen}{c, c, c} \right\}$
    };

    \def\barH{0.6}  
    \def\yPos{0}    
    \def\gap{0} 

    \filldraw[fill=sketchRed, draw=black, line width=1pt] 
        (0, \yPos) rectangle ++(4.2, \barH);
    \node[text=sketchRed, below=5pt] at (2.1, \yPos) {$a$};
    \node[text=sketchRed, above = 20pt] at (2.1, \yPos) {\normalsize $X_1, X_2 \sim I_a$};

    \filldraw[fill=sketchOrange, draw=black, line width=1pt] 
        (4.2 + \gap, \yPos) rectangle ++(2.0, \barH);
    \node[text=sketchOrange, below=5pt] at (4.2 + \gap + 1, \yPos+0.135) {$b$};
    \node[text=sketchOrange, above=20pt] at (4.2 + \gap + 1, \yPos) {\normalsize $X_3 \sim I_b$};

    \filldraw[fill=sketchGreen, draw=black, line width=1pt] 
        (6.2 + 2*\gap, \yPos) rectangle ++(3.0, \barH);
    \node[text=sketchGreen, below=5pt] at (6.2 + 2*\gap + 1.5, \yPos+0.01) {$c$};
    \node[text=sketchGreen, above=20pt] at (6.2 + 2*\gap + 1.5, \yPos) {\normalsize $X_4, X_5, X_6, \textcolor{black}{Z} \sim I_c$};

    \filldraw[fill=sketchBlue, draw=black, line width=1pt] 
        (9.2 + 3*\gap, \yPos) rectangle ++(2.5, \barH);
    \node[text=sketchBlue, below=5pt] at (9.2 + 3*\gap + 1.25, \yPos+0.13) {$d$};

\end{tikzpicture}
   \caption{Partition of the interval $[0, 1]$ induced by the preference order $d \succ_v c \succ_v b \succ_v a$. The multiset is $S = \set{a, a, b, c, c, c}$. To compute $\Pr{}{c \succ_v S}$, we sample $X_1, \dots, X_6$ and $Z$ independently from the corresponding subintervals of $[0, 1]$. In this example, we have $\Pr{}{c \succ_v S} = \Pr{}{\max\set{X_1, \dots, X_6} < Z} = 1/4.$}
  \label{fig:interval}
\end{figure}

We now relate the comparison between a candidate $a$ and a multiset $S = \set{s_1, \dots, s_k}$ to this interval representation in a way that is consistent with our tie-breaking model (see \cref{fig:interval}). Fix a voter $v$. Consider sample a point $Z$ uniformly from $I_{a, v}$ and a point $X_{i, v}$ uniformly from $I_{s_i, v}$ for each $i \in [k]$. All random variables are mutually independent.

If $a$ is the most preferred element of $\set{a} \cup S$ according to $v$, then the largest among $\set{Z, X_{1, v}, \dots, X_{k, v}}$ is equally likely to be any of these $n_S(a) + 1$ points. Therefore, the probability of $Z$ being the largest one is exactly $1/(n_S(a) + 1)$. Otherwise, $a$ is dominated by some candidate in $S$, in which case $Z$ cannot be the largest point and the probability is $0$.

Consequently, we have
\begin{equation}
\Pr{}{a \succ_v S} = \Pr{\substack{Z \sim I_{a, v}\\X_{i, v} \sim I_{s_i, v} \text{ for all }i \in [k]}}{\max_{i \in [k]}\set{X_{i, v}} < Z} = \E{Z \sim I_{a, v}}{\prod_{i = 1}^k\Pr{X_{i, v} \sim I_{s_i,v}}{X_{i, v} < Z}}.\label{eq:multiset_cmp}
\end{equation}

 We now independently sample the candidates $c_1, \dots, c_q \sim D$ with replacement to form the empirical distribution $\hat{D}$. For each voter $v$, the interval construction associates each sampled candidate $c_i$ with a random point $X_{i,v}\sim I_{c_i,v}$. By considering the conditional probability of choosing any value of $X_{i, v}$, we see that from the point of view of each voter, each $X_{i,v}$ can be viewed as a uniform random sample from $[0, 1]$. This sampling process ensures that we can sample one multiset and compare samples from that multiset with another candidate for \emph{all} voters simultaneously. This method of associating sampled candidates with sampled points will be formalized in our proof of \cref{thm:repapx_SL_existence} and is key to establish our desired bounds.

\subsection{Proof of \texorpdfstring{\cref{thm:repapx_SL_existence}}{a Theorem}}
\begin{proof}[Proof of \cref{thm:repapx_SL_existence}]
With our interval construction in place, we are now able to show the existence of $q$ samples that satisfy \cref{thm:repapx_SL_existence}. We will also explicitly calculate $q$ in terms of $\veps$. 

Let $D$ be a Stable $k$-Lottery $D_\kSL{k}$. Let $\widehat{D}$ denote a distribution that chooses uniformly at random from a multiset of candidates $\{c_1,\dots,c_q\}$. The following definitions will assume that the multiset $\{c_1,c_2,\ldots,c_q\}$ has already been fixed, but the remainder of the proof deal with random multisets formed by $q$ random samples from $D$. By \cref{eq:multiset_cmp}, our sampling method gives us that the winning probability of candidate $a$ over $\widehat{D}^k$ from the view of voter $v$ is
\begin{equation}
  \begin{split}
	&\Pr{}{a \succ_v \widehat{D}^k} = \E{Z \sim I_{a,v}}{\E{s_1, \dots, s_k \sim \widehat{D}}{\prod_{i=1}^k\Pr{X_{i, v}\sim I_{s_i, v}}{X_{i, v} < Z}}} \\
    ={}&{}\E{Z \sim I_{a,v}}{\left(\Pr{\substack{c \sim \widehat{D}\\X \sim I_{c,v}}}{X < Z}\right) ^ k} =  \E{Z \sim I_{a,v}}{\left(\E{X_1 \sim I_{c_1,v}, \dots, X_q \sim I_{c_q,v}}{\frac1q\sum_{i = 1}^q\mathds{1}[X_i < Z]}\right) ^ k},
  \end{split}\label{eq:pr_a_beats_d}
\end{equation}
while the winning probability of candidate $a$ over $D$ from the view of voter $v$ is
\begin{equation}
  \begin{split}
	&\Pr{}{a \succ_v D^k} = \E{Z \sim I_{a,v}}{\E{s_1, \dots, s_k \sim D}{\prod_{i=1}^k\Pr{X_{i, v}\sim I_{s_i, v}}{X_{i, v} < Z}}} \\={}&{} \E{Z \sim {I}_{a,v}}{\left(\Pr{X \sim [0, 1]}{X < Z}\right) ^ k}= \E{Z \sim {I}_{a,v}}{Z ^ k}.
  \end{split}\label{eq:pr_a_beats_dslk}
\end{equation}

Notably, this sampling procedure preserves the fact that with respect to each voter, samples from $D$ can be viewed as uniform samples from the unit interval. We denote the maximum difference, for a fixed $v$ over all possible $a\in C$, between \cref{eq:pr_a_beats_d,eq:pr_a_beats_dslk} as
\[
\delta_v \defeq \max_{a\in C} \left(\Pr{}{a \succ_v \widehat{D}^k} - \Pr{}{a \succ_v D^k}\right).
\]
First, we show that for any voter, the expected difference, $\E{c_1, \dots, c_q \sim D}{\delta_v}$, is not too large. 
\begin{align*}
	& \E{c_1, \dots, c_q \sim D}{\delta_v} \\
    = {}&{}\E{c_1, \dots, c_q \sim D}{\max_{a\in C}\left(\E{Z \sim I_{a,v}}{\left(\E{X_1 \sim I_{c_1,v}, \dots, X_q \sim I_{c_q,v}}{\frac1q\sum_{i = 1}^q\mathds{1}[X_i < Z]}\right) ^ k - Z^k}\right)} \tag{\cref{eq:pr_a_beats_d,eq:pr_a_beats_dslk}}\\
    \leq {}&{}\E{c_1, \dots, c_q \sim D}{\max_{a\in C}\sup_{r_a \in I_{a,v}}\left({\left(\E{X_1 \sim I_{c_1,v}, \dots, X_q \sim I_{c_q,v}}{\frac1q\sum_{i = 1}^q\mathds{1}[X_i < r_a]}\right) ^ k - r_a^k}\right)}\tag{$\E{}{f} \leq \sup f$}  \\
	= {}&{} \E{c_1, \dots, c_q \sim D}{\sup_{r\in[0, 1]}\left(\left( \E{X_1 \sim I_{c_1,v}, \dots, X_q \sim I_{c_q,v}}{\frac1q\sum_{i = 1}^q\mathds{1}[X_i < r]}\right) ^ k - r ^ k \right)} \\
	\leq {}&{} k \cdot \E{c_1, \dots, c_q \sim D}{{\sup_{r\in[0, 1]}\abs{\E{X_1 \sim I_{c_1,v}, \dots, X_q \sim I_{c_q,v}}{\frac1q\sum_{i = 1}^q\mathds{1}[X_i < r] }-r}}} \tag{$x ^ k - y ^ k \leq k \cdot \abs{x - y}$}\\
	\leq {}&{} k \cdot \E{c_1, \dots, c_q \sim D}{{\sup_{r\in[0, 1]}\E{X_1 \sim I_{c_1,v}, \dots, X_q \sim I_{c_q,v}}{\abs{\frac1q\sum_{i = 1}^q\mathds{1}[X_i < r]-r}}}} \tag{$\abs{\E{}{f}} \leq \E{}{\abs{f}}$}\\
	\leq {}&{} k \cdot \E{c_1, \dots, c_q \sim D}{\E{X_1 \sim I_{c_1,v}, \dots, X_q \sim I_{c_q,v}}{\sup_{r\in[0, 1]}\abs{\frac1q\sum_{i = 1}^q\mathds{1}[X_i < r] - r}}} \tag{$\sup\E{}{f} \leq \E{}{\sup f}$} \\
	= {}&{} k \cdot \E{X_1, \dots, X_q \sim [0, 1]}{\sup_{r \in [0, 1]}\abs{\frac 1q\sum_{i = 1}^q\mathds{1}[X_i < r] - r}}\\
	\leq {}&{} k \cdot d(q) \tag{\cref{cor:DKW-target}}.
\end{align*}

Linearity of expectation now allows us to show that the above bound holds for the \emph{average} voter as well.  
\begin{equation}
\E{c_1, \dots, c_q \sim D}{\E{v\sim V}{\delta_v}} = \E{v\sim V}{\E{c_1, \dots, c_q \sim D}{\delta_v}} \leq \E{v\sim V}{k\cdot d(q)} = k\cdot d(q) \leq k \sqrt{\frac{\pi}{2q}}. \label{eq:stable_lottery_expectation}
\end{equation}
In particular, this means that there is a positive probability of choosing $q$ samples $c_1,\ldots,c_q$ such that $\widehat{D}$, the uniform distribution over the multiset $\{c_1,\ldots, c_q\}$, satisfies
\[
\frac{1}{n}\sum_{v\in V}\left(\Pr{}{a \succ_v \widehat{D}^k}-\Pr{}{a \succ_v D^k}\right)\leq k\sqrt{\frac{\pi}{2q}}
\]
for all candidate $a\in C$. This same distribution $\widehat{D}$ must satisfy
\[
\Pr{v\sim V}{a\succ _v \widehat{D}^k}=\frac{1}{n}\sum_{v\in V}\Pr{}{a\succ_v \widehat{D}^k}\leq \frac{1}{n}\sum_{v\in V}\left(\Pr{}{a\succ_v D^k}+k\sqrt{\frac{\pi}{2q}}\right)\leq \frac{1}{k+1}+k\sqrt{\frac{\pi}{2q}}
\]
for all $a\in C$ as well, and hence is a $\left(k\sqrt{\frac{\pi}{2q}}\right)$-RepApx Stable $k$-Lottery. 

Therefore, for a fixed $k$ and desired approximation error $\veps$, it suffices to have $\frac{\pi k^2}{2\veps^2}$ samples from a Stable $k$-Lottery $D_{\kSL{k}}$. 
\end{proof}

Although we present \cref{thm:repapx_SL_existence} as an existence result, we can turn our result into an algorithmic one with a small overhead. Note that we can apply Markov's inequality to the expectation bound of \cref{eq:stable_lottery_expectation} and derive that if we draw $q$ random samples, the uniform distribution over these samples has a probability of at least $\frac{\gamma}{1+\gamma}$ to be a $\left((1+\gamma)k\sqrt{\frac{\pi}{2q}}\right)$-Stable $k$-Lottery. We can therefore repeat this process until success. See also \cite[Remark 2]{DBLP:conf/soda/CharikarRW26}.

\section{Support of Maximal Lotteries Has Constant Distortion}
\label{sec:ml_support}
In this section, we prove that every candidate in the support of Maximal Lotteries (i.e., every candidate selected with positive probability by the Maximal Lotteries rule) has metric distortion of at most $4 + \sqrt{17}$. Along the way, we establish several properties in \cref{sec:tournament} that relate local pairwise-comparison patterns to upper bounds on the social-cost ratio between two candidates.

We first state the main theorem of this section.
\begin{theorem}[Deterministic Selection in the Support of Maximal Lotteries]
\label{thm:const_dist_supp_ML}
  In any election, for any candidate $j^* \in \Supp(D_\ML)$, the distortion of $j^*$ is at most $4 + \sqrt{17} < 8.124$.
\end{theorem}

To prove \cref{thm:const_dist_supp_ML}, we need the following two simple facts about Maximal Lotteries.
\begin{fact}[see, e.g., {\cite[Claim 1]{DBLP:journals/jacm/CharikarRWW24}}]
\label[fact]{fac:ML_comp}
For any $I \subseteq C$ with $0 < P(I) < 1$, it holds that
\[
s_{D_\ML(I) \succ D_\ML(I^c)} = s_{D_\ML(I^c) \succ D_\ML(I)} = \frac12.
\]
\end{fact}
\begin{fact}
\label[fact]{fac:support_vs_ML}
For any $j \in \Supp(D_\ML)$, it holds that
\[
s_{D_\ML \succ j} = s_{j \succ D_\ML} = \frac12.
\]
\end{fact}
\begin{proof}[Proof of \cref{fac:support_vs_ML}]
If $p_j = 1$, the equality immediately holds since $D_\ML$ selects $j$ deterministically. Otherwise, we compute
\[
s_{j \succ D_\ML} = s_{j \succ j}p_j + s_{j \succ D_\ML(C \setminus\set{j})} \cdot \left(1-p_j\right)= \frac{p_j}2 + \frac{1-p_j}2 = \frac12,
\]
where the first equality follows from the law of total probability and the second follows from \cref{fac:ML_comp}.
\end{proof}

A key step in the proof of \cref{thm:const_dist_supp_ML} is the following lemma.
\begin{lemma}\label{lem:supp_ML_two_steps}
  In any election, for any candidate $j^* \in \Supp(D_\ML)$ and any constant $\theta \in \left(0, \frac12\right)$, there exists a candidate $k$ such that $s_{j^* \succ k} \geq \theta$ and $s_{k \succ i^*} \geq \frac12 - \theta$.
\end{lemma}

\begin{proof}[Proof of \cref{lem:supp_ML_two_steps}]
Let $D = D_\ML$ for convenience. Assume for contradiction that no such candidate $k$ exists. Define the following two sets of candidates:
\[
X \defeq \set{k \in C:s_{k \succ i^*} \geq \frac12 - \theta} ~~~\text{ and }~~~ Y \defeq \set{k \in C:s_{j ^ * \succ k} \geq \theta}.
\]
By our assumption, these sets are disjoint. Our goal is to show that both $P(X)$ and $P(Y)$ are sufficiently large so that their total probability exceeds $1$, yielding a contradiction.

\paragraph{Proving $P(Y)$ is large.}
By the definition of $Y$, for any candidate $k \in C \setminus Y$, the probability that $j^*$ beats $k$ is less than $\theta < \frac{1}{2}$. However, since the candidate $j^*$ has positive support, we know $s_{j^* \succ D} \geq \frac12$ from \cref{fac:support_vs_ML}. Therefore, the probability that $D$ selects from $Y$ should be large.

We now formalize this intuition. First, we claim $P(Y) \neq 0$. Otherwise, (if $P(Y) = 0$), we have $D = D(C \setminus Y)$, and thus
\[
s_{j ^ * \succ D} = s_{j ^ * \succ D(C \setminus Y)} \mkern-10mu \stackrel{ \text{\cref{fac:set2dis}} }{\leq} \mkern-10mu \max_{k \notin Y}{s_{j^* \succ k}} < \theta < \frac 12,
\]
which contradicts \cref{fac:support_vs_ML}. Next, if $P(C \setminus Y) = 0$, then $P(Y) = 1$, which is good for our purpose of showing that $P(Y)$ is large. It therefore remains to consider the nontrivial case where $0 < P(Y) < 1$. In this case, both $D(Y)$ and $D(C \setminus Y)$ are well-defined. We can get
\begin{align*}
  \frac12 \leq s_{j^* \succ D} &= \bigl(1 - P(Y)\bigr) \cdot s_{j^* \succ D(C \setminus Y)} + P(Y) \cdot s_{j^* \succ D(Y)} \tag{law of total probability}\\
  &\leq \bigl(1 - P(Y)\bigr) \cdot s_{j^* \succ D(C \setminus Y)} + P(Y) \cdot \left(s_{j^* \succ D(C \setminus Y)} + s_{D(C \setminus Y) \succ D(Y)}\right)\tag{\cref{lem:triangle}}\\
  &= s_{j^* \succ D(C\setminus Y)} + P(Y) \cdot s_{D(C \setminus Y) \succ D(Y)}\\
  &< \theta + \frac{P(Y)}2. \tag{\cref{fac:set2dis,fac:ML_comp}}
\end{align*} Rearranging yields $P(Y) > 1 - 2\theta$.

\paragraph{Proving $P(X)$ is large.}
By the definition of $X$, for any candidate $k \in C \setminus X$, the probability that $k$ beats $i^*$ is less than $\frac{1}{2} - \theta < \frac{1}{2}$. However, by \cref{thm:ML_distribution}, we know $s_{D \succ i^*} \geq \frac12$. Therefore, the probability that $D$ selects from $X$ should be large.

The intuition is formalized as follows. We first rule out the case $P(X) = 0$. Indeed, if $P(X) = 0$, then $D = D(C \setminus X)$ and we have
\[
s_{D \succ i^*} = s_{D(C \setminus X) \succ i ^*} \mkern-10mu \stackrel{ \text{\cref{fac:set2dis}} }{\leq} \mkern-10mu \max_{k \notin X} s_{k \succ i^*} < \frac12 - \theta < \frac12,
\]
contradicting \cref{thm:ML_distribution}. Next, if $P(C \setminus X) = 0$, then $P(X) = 1$, good for our purpose of showing $P(X)$ is large. Hence, it remains to consider the nontrivial case $0 < P(X) < 1$. Here, both $D(X)$ and $D(C\setminus X)$ are well-defined, and we get
\begin{align*}
  \frac12 \leq s_{D \succ i^*} &= \bigl(1 - P(X)\bigr) \cdot s_{D(C \setminus X) \succ i^*} + P(X) \cdot s_{D(X) \succ i^*} \tag{law of total probability}\\
  &\leq \bigl(1 - P(X)\bigr) \cdot s_{D(C \setminus X) \succ i^*} + P(X) \cdot \left(s_{D(X) \succ D(C \setminus X)} + s_{D(C \setminus X) \succ i^*}\right)\tag{\cref{lem:triangle}}\\
  &= s_{D(C\setminus X) \succ i^*} + P(X) \cdot s_{D(X) \succ D(C \setminus X)}\\
  &< \frac12 - \theta + \frac{P(X)}2.\tag{\cref{fac:set2dis,fac:ML_comp}}
\end{align*} Rearranging yields $P(X) > 2\theta$.

Combining the inequalities of $P(Y) > 1 - 2\theta$ and $P(X) > 2\theta$ gives a contradiction: $P(X) + P(Y) > 1$. Therefore, the desired candidate $k$ must exist, completing the proof.
\end{proof}

In \cref{sec:tournament}, we will state several lemmas (\cref{lem:one-hop,lem:two-hop-var,lem:two-hop-balance}) that bound the ratio of social costs between two candidates, given their local relationships in the weighted tournament graph (which summarizes all information of $s_{j \succ k}$ for all pairs of candidates $j, k \in C$). In particular, \cref{lem:two-hop-balance} (which is a bound for $\SC(j^*) / \SC(i^*)$ given conditions for $s_{j^* \succ k}$ and $s_{k \succ i^*}$) can be combined with \cref{lem:supp_ML_two_steps} to immediately complete the proof of \cref{thm:const_dist_supp_ML}. We now present the proof of \cref{thm:const_dist_supp_ML} assuming \cref{lem:two-hop-balance}, and defer the proof of \cref{lem:two-hop-balance} to \cref{sec:tournament}.

\begin{proof}[Proof of \cref{thm:const_dist_supp_ML}]
  By \cref{lem:supp_ML_two_steps}, for any constant $\theta \in \left(0, \frac12\right)$, there exists a candidate $k$ such that $s_{j^* \succ k}\geq \theta$ and $s_{k \succ i^*} \geq \frac12 - \theta$. Applying \cref{lem:two-hop-balance}, the distortion of $j^*$ is therefore at most the minimum value of
\[
1 + 2 \max\set{\frac1{\frac12 - \theta}, \frac{1 - \theta}{\theta}} \qquad \text{subject to $\theta \in \left(0, \frac12\right)$.}
\]
Taking $\theta = \frac{5-\sqrt{17}}4$ minimizes the above expression, and gives a distortion of $4 + \sqrt{17}$.
\end{proof}

\subsection{Ratio of Social Costs in a Weighted Tournament Graph}
\label{sec:tournament}

We will need three simple lemmas which bound the distortion of candidates using local conditions on the weights of edges in the tournament graph. These can all be viewed as corollaries of the following result of \cite{DBLP:conf/sigecom/CharikarRTW25}.

\begin{lemma}[{\cite[Corollary 4.3]{DBLP:conf/sigecom/CharikarRTW25}}]\label[lemma]{lem:two-hop-general}
$\SC(j) \leq (1 + 2\lambda)\SC(i^*)$ if candidate $j$ and $k$ satisfy that \[s_{i^* \succ j} \leq \lambda s_{k \succ i^*}\] and for all partitions of the candidates $I \sqcup J = C$ such that $i^*, k \in I$ and $j \in J$, it holds
\[
\min_{i \in I}s_{i \succ j} \leq \lambda\max_{\ell \in J}\set{s_{\ell \succ i^*}, s_{\ell \succ k}}.
\]
\end{lemma}

The first two lemmas are already proven explicitly in prior work. Proofs for both can be found in \cite{DBLP:conf/sigecom/CharikarRTW25}).

\begin{lemma}[Implicit in \cite{DBLP:conf/ec/MunagalaW19} and Corollary of {\cite[Corollary 5.1]{DBLP:conf/aaai/000120a}}]
\label[lemma]{lem:one-hop}
If a candidate $j$ satisfies that $s_{j \succ i^*} \geq \theta$ for some $\theta \in (0, 1]$, then
\[
\SC(j) \leq \left(\frac 2\theta - 1\right) \cdot \SC(i^*).
\]
\end{lemma}

\begin{lemma}[{\cite[Corollary 4.4]{DBLP:conf/sigecom/CharikarRTW25}}]
\label[lemma]{lem:two-hop-var}
If candidates $j$ and $k$ satisfy that $s_{j \succ k} \geq \theta$ and $s_{k \succ i^*} \geq \theta$ for some $\theta \in (0, 1]$, then
\[
\SC(j) \leq \left(\frac 4\theta - 3\right) \cdot \SC(i^*).
\]
\end{lemma}

\begin{lemma}
\label[lemma]{lem:two-hop-balance}
If candidates $j$ and $k$ satisfy that $s_{j \succ k} \geq \theta$ and $s_{k \succ i^*} \geq \frac12 - \theta$ for some $\theta \in (0, \frac{1}{2})$, then
\[
\SC(j) \leq \left(1 + 2 \max\set{\frac1{\frac12 - \theta}, \frac{1 - \theta}{\theta}}\right) \cdot \SC(i^*).
\]
\end{lemma}

\begin{proof}[Proof of \cref{lem:two-hop-balance}]
On the one hand, if $\lambda \geq \frac1{\frac12 - \theta}$, then we have
\[
s_{i^* \succ j} \leq 1 \leq \left(\frac12 - \theta\right)\cdot \lambda \leq\lambda s_{k \succ i^*}.
\]

On the other hand, if $\lambda \geq \frac{1 - \theta}{\theta}$, then we have
\[
\min_{i \in I}{s_{i \succ j}} \mkern-3mu \stackrel{k \in I}{\leq} \mkern-3mu  s_{k \succ j} \leq 1 - \theta \leq \lambda \theta \leq \lambda s_{j \succ k}  \mkern-3mu \stackrel{j \in J}{\leq} \mkern-3mu  \lambda\max_{\ell \in J}\set{s_{\ell \succ i^*}, s_{\ell \succ k}}. 
\]

Therefore, by \cref{lem:two-hop-general}, for any $\lambda \geq \max\set{\frac1{\frac12 - \theta}, \frac{1 - \theta}{\theta}}$, we have $\SC(j) \leq (1 + 2\lambda)\SC(i^*)$. This completes the proof.
\end{proof}

\section{RepApx Maximal Lotteries Get Distortion Near \texorpdfstring{$3$}{3}}
\label{sec:ml_distortion}
In this section, we prove that the metric distortion of an $\veps^2$-RepApx Maximal Lottery is $3 + O(\veps)$.

\begin{theorem}[Metric Distortion of $\veps^2$-RepApx Maximal Lotteries]\label{thm:distortion_repapx_ML}In any election and for any $\veps < 1/\sqrt{10}$, an $\veps ^ 2$-RepApx Maximal Lottery $D_\eML{\veps ^ 2}$ has metric distortion at most $3 + 28\veps$.
\end{theorem}

In their proof that Maximal Lotteries have metric distortion of at most $3$, \cite{DBLP:journals/jacm/CharikarRWW24} establish a relation that $r(t)$ is never below $\ell(D_\ML, t)$ (see \cref{fig:apx_ML_fig1}). Therefore, taking the integral over $t \geq 0$ yields $L(D_\ML) \leq R$. By \cref{thm:biased_metric}, the metric distortion of Maximal Lotteries is at most $3$.

\begin{figure}[t!]
    \centering
    \begin{subfigure}[t]{0.49\textwidth}
        \centering
        \begin{tikzpicture}[
    scale=0.65,
myredline/.style={color=red!80!black, line width=1pt},
    myredfill/.style={color=red!60, opacity=0.8},
    myblueline/.style={color=blue!40!teal!80!black, line width=1pt},
    mybluefill/.style={color=blue!40!teal!30, opacity=0.8},
    axis/.style={line width=1pt, black},
    dashed guide/.style={dashed, thin, black!80, dash pattern=on 5pt off 3pt},
dot/.style={circle, fill=white, draw=#1, line width=1pt, inner sep=1pt},
filldot/.style={circle, fill=#1, inner sep=1.3pt}
]

\def\yTop{7}
    \def\yB{4.2}
    \def\yHalf{3.5}
    \def\yZero{0}
    \def\xMax{9}
    \def\xTau{4.0}
    \def\xAlphaR{5.0}
    \def\xMu{5.9}

    \begin{scope}
\fill[myredfill] 
            (0,0) -- (0,6.2) -- (2.1,6.2) -- (2.1,5.5) -- (3.3,5.5) -- 
            (3.3,4.5) -- (5.2,4.5) -- (5.2,4.0) -- (5.9, 4.0) -- 
            (5.9,2.6) -- (7.8,2.6) -- (7.8,0) -- (\xMax,0) -- cycle;
    \end{scope}

    \fill[mybluefill] 
        (0,\yZero) -- (0,2.8) -- (2.3,2.8) -- (2.3,2.0) -- (4.0,2.0) -- 
        (4.0,1.3) -- (6.3,1.3) -- (6.3,0.8) -- (7,0.8) -- (7,\yZero) -- cycle;

\draw[axis] (0, \yTop) -- (0, 0) -- (\xMax, 0) node[below, xshift=-0.1cm] {$t$};

    \draw[line width=1pt] (0, \yTop - 0.03) -- (-0.2, \yTop - 0.03) node[left] {$1$};

\coordinate (B_start) at (0, 2.8);
    \coordinate (B1_R) at (2.3, 2.8); \coordinate (B1_low) at (2.3, 2.0);
    \coordinate (B2_R) at (4.0, 2.0); \coordinate (B2_low) at (4.0, 1.3);
    \coordinate (B3_R) at (6.3, 1.3); \coordinate (B3_low) at (6.3, 0.8);
    \coordinate (B4_R) at (7, 0.8); \coordinate (B4_low) at (7, 0);
    \coordinate (B_end) at (\xMax, 0);

    \draw[myblueline] (0, 2.8) -- (B1_R);
    \draw[myblueline] (B1_low) -- (B2_R);
    \draw[myblueline] (B2_low) -- (B3_R);
    \draw[myblueline] (B3_low) -- (B4_R);
    \draw[myblueline] (B4_low) -- (B_end);

    \draw[myblueline, dashed] (B1_R) -- (B1_low);
    \draw[myblueline, dashed] (B2_R) -- (B2_low);
    \draw[myblueline, dashed] (B3_R) -- (B3_low);
    \draw[myblueline, dashed] (B4_R) -- (B4_low);

\def\blueColor{blue!40!teal!80!black}
    
    \node[dot=\blueColor] at (B1_R) {};
    \node[filldot=\blueColor] at (B1_low) {};
    
    \node[dot=\blueColor] at (B2_R) {};
    \node[filldot=\blueColor] at (B2_low) {};
    
    \node[dot=\blueColor] at (B3_R) {};
    \node[filldot=\blueColor] at (B3_low) {};
    
    \node[dot=\blueColor] at (B4_R) {};
    \node[filldot=\blueColor] at (B4_low) {};

\coordinate (R_start) at (0, 6.2);
    \coordinate (R1_R) at (2.1, 6.2); \coordinate (R1_low) at (2.1, 5.5);
    \coordinate (R2_R) at (3.3, 5.5); \coordinate (R2_low) at (3.3, 4.5);
    \coordinate (R3_R) at (5.2, 4.5); \coordinate (R3_low) at (5.2, 4.0);
    \coordinate (R4_R) at (5.9, 4.0); \coordinate (R4_low) at (5.9, 2.6);
    \coordinate (R5_R) at (7.8, 2.6); \coordinate (R5_low) at (7.8, 0);
    \coordinate (R_end) at (\xMax, 0);

    \draw[myredline] (0, 6.2) -- (R1_R);
    \draw[myredline] (R1_low) -- (R2_R);
    \draw[myredline] (R2_low) -- (R3_R);
    \draw[myredline] (R3_low) -- (R4_R);
    \draw[myredline] (R4_low) -- (R5_R);
    \draw[myredline] (R5_low) -- (R_end);

    \draw[myredline, dashed] (R1_R) -- (R1_low);
    \draw[myredline, dashed] (R2_R) -- (R2_low);
    \draw[myredline, dashed] (R3_R) -- (R3_low);
    \draw[myredline, dashed] (R4_R) -- (R4_low);
    \draw[myredline, dashed] (R5_R) -- (R5_low);

\def\redColor{red!80!black}

    \node[dot=\redColor] at (R1_R) {};
    \node[filldot=\redColor] at (R1_low) {};
    \node[dot=\redColor] at (R2_R) {};
    \node[filldot=\redColor] at (R2_low) {};
    \node[dot=\redColor] at (R3_R) {};
    \node[filldot=\redColor] at (R3_low) {};
    \node[dot=\redColor] at (R4_R) {};
    \node[filldot=\redColor] at (R4_low) {};
    \node[dot=\redColor] at (R5_R) {};
    \node[filldot=\redColor] at (R5_low) {};

\node[blue!40!teal!80!black, anchor=west] at (4.5, 1.8) {$\ell(D_\ML, t)$};
    
\node[red!80!black, anchor=south west] at (6.4, 2.6) {$r(t)$};

\end{tikzpicture}

         \caption{$\ell(D_\ML, t)$ v.s.\@ $r(t)$}
        \label{fig:apx_ML_fig1}
    \end{subfigure}~ 
    \begin{subfigure}[t]{0.49\textwidth}
        \centering
        \begin{tikzpicture}[
    scale=0.65,
myredline/.style={color=red!80!black, line width=1pt},
    myredfill/.style={color=red!60, opacity=0.8},
    myblueline/.style={color=blue!40!teal!80!black, line width=1pt},
    mybluefill/.style={color=blue!40!teal!30, opacity=0.8},
    axis/.style={line width=1pt, black},
    dashed guide/.style={dashed, thin, black!80, dash pattern=on 5pt off 3pt},
dot/.style={circle, fill=white, draw=#1, line width=1pt, inner sep=1pt},
filldot/.style={circle, fill=#1, inner sep=1.3pt}
]

\def\yTop{7}
    \def\yB{4.2}
    \def\yHalf{3.5}
    \def\yZero{0}
    \def\xMax{9}
    \def\xTau{4.0}
    \def\xAlphaR{5.0}
    \def\xMu{5.9}
    \def\yEps{0.2}

    \begin{scope}
\fill[myredfill] 
            (0,0) -- (0,6.2) -- (1.8,6.2) -- (1.8,5.5) -- (3.5, 5.5) -- 
            (3.5,4.5) -- (5.5,4.5) -- (5.5,1.5) -- (7.0, 1.5) -- 
            (7.0,1.0) -- (8.5,1.0) -- (8.5,\yEps) -- (\xMax,\yEps) -- (\xMax, 0) -- cycle;
    \end{scope}

    \fill[mybluefill] 
        (0,\yZero) -- (0,3.2) -- (2.2,3.2) -- (2.2,2.8) -- (4.0,2.8) -- 
        (4.0,2.4) -- (6.3,2.4) -- (6.3,1.9) -- (8,1.9) -- (8,\yZero) -- cycle;

\draw[axis] (0, \yTop) -- (0, 0) -- (\xMax, 0) node[below] {$t$};

    \draw[line width=1pt] (0, \yTop - 0.03) -- (-0.2, \yTop - 0.03) node[left] {$1 + 2\veps$};
\draw[line width=1pt] (0, \yEps - 0.03) -- (-0.2, \yEps - 0.03) node[left] {$2\veps$};

\coordinate (B_start) at (0, 3.2);
    \coordinate (B1_R) at (2.2, 3.2); \coordinate (B1_low) at (2.2, 2.8);
    \coordinate (B2_R) at (4.0, 2.8); \coordinate (B2_low) at (4.0, 2.4);
    \coordinate (B3_R) at (6.3, 2.4); \coordinate (B3_low) at (6.3, 1.9);
    \coordinate (B4_R) at (8, 1.9); \coordinate (B4_low) at (8, 0);
    \coordinate (B_end) at (\xMax, 0);

    \draw[myblueline] (B_start) -- (B1_R);
    \draw[myblueline] (B1_low) -- (B2_R);
    \draw[myblueline] (B2_low) -- (B3_R);
    \draw[myblueline] (B3_low) -- (B4_R);
    \draw[myblueline] (B4_low) -- (B_end);

    \draw[myblueline, dashed] (B1_R) -- (B1_low);
    \draw[myblueline, dashed] (B2_R) -- (B2_low);
    \draw[myblueline, dashed] (B3_R) -- (B3_low);
    \draw[myblueline, dashed] (B4_R) -- (B4_low);

\def\blueColor{blue!40!teal!80!black}
    
    \node[dot=\blueColor] at (B1_R) {};
    \node[filldot=\blueColor] at (B1_low) {};
    
    \node[dot=\blueColor] at (B2_R) {};
    \node[filldot=\blueColor] at (B2_low) {};
    
    \node[dot=\blueColor] at (B3_R) {};
    \node[filldot=\blueColor] at (B3_low) {};
    
    \node[dot=\blueColor] at (B4_R) {};
    \node[filldot=\blueColor] at (B4_low) {};

\coordinate (R_start) at (0, 6.2);
    \coordinate (R1_R) at (1.8, 6.2); \coordinate (R1_low) at (1.8, 5.5);
    \coordinate (R2_R) at (3.5, 5.5); \coordinate (R2_low) at (3.5, 4.5);
    \coordinate (R3_R) at (5.5, 4.5); \coordinate (R3_low) at (5.5, 1.5);
    \coordinate (R4_R) at (7.0, 1.5); \coordinate (R4_low) at (7.0, 1.0);
    \coordinate (R5_R) at (8.5, 1.0); \coordinate (R5_low) at (8.5, \yEps);
    \coordinate (R6_R) at (\xMax, \yEps); \coordinate (R_end) at (\xMax, 0);

    \draw[myredline] (R_start) -- (R1_R);
    \draw[myredline] (R1_low) -- (R2_R);
    \draw[myredline] (R2_low) -- (R3_R);
    \draw[myredline] (R3_low) -- (R4_R);
    \draw[myredline] (R4_low) -- (R5_R);
    \draw[myredline] (R5_low) -- (R6_R);

    \draw[myredline, dashed] (R1_R) -- (R1_low);
    \draw[myredline, dashed] (R2_R) -- (R2_low);
    \draw[myredline, dashed] (R3_R) -- (R3_low);
    \draw[myredline, dashed] (R4_R) -- (R4_low);
    \draw[myredline, dashed] (R5_R) -- (R5_low);

\def\redColor{red!80!black}

    \node[dot=\redColor] at (R1_R) {};
    \node[filldot=\redColor] at (R1_low) {};
    \node[dot=\redColor] at (R2_R) {};
    \node[filldot=\redColor] at (R2_low) {};
    \node[dot=\redColor] at (R3_R) {};
    \node[filldot=\redColor] at (R3_low) {};
    \node[dot=\redColor] at (R4_R) {};
    \node[filldot=\redColor] at (R4_low) {};
    \node[dot=\redColor] at (R5_R) {};
    \node[filldot=\redColor] at (R5_low) {};

\node[blue!40!teal!80!black, anchor=west] at (5.5, 3) {$\ell(D_\eML{\veps^2}, t)$};
    
\node[red!80!black, anchor=south west] at (3.6, 4.6) {$r(t) + 2\veps$};
\end{tikzpicture}
         \caption{$\ell(D_\eML{\veps^2}, t)$ v.s.\@ $r(t) + 2\veps$}
        \label{fig:apx_ML_fig2}
    \end{subfigure}
    \caption{Pictorial illustration of $\ell(D_\ML, t)$ v.s.\@ $r(t)$, and $\ell(D_\eML{\veps^2}, t)$ v.s.\@ $r(t) + 2\veps$}
\end{figure}
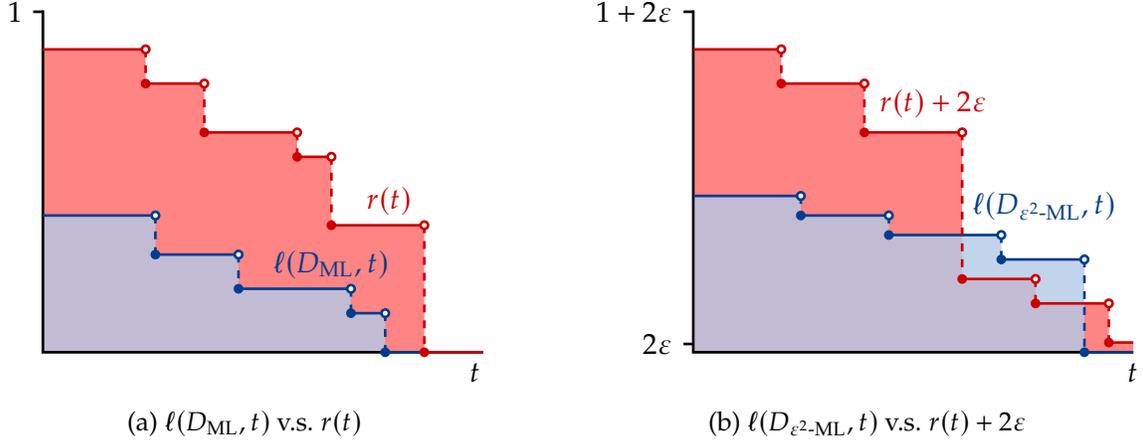

However, for the approximate version, $\veps^2$-RegApx Maximal Lotteries, the pointwise inequality (i.e., $\ell(D_\ML, t) \leq r(t)$ for all $t \geq 0$) is no longer guaranteed. A natural first attempt is to ``lift'' the red curve $r(t)$ up by some amount (e.g., $2\veps$) to keep its plot above the blue one $\ell(D_\eML{\veps^2})$. However, it turns out that such an approach is insufficient: even after the lift, the red curve can still fall below the blue one (\cref{fig:apx_ML_fig2}). What we will do instead is perform a careful case analysis to resolve this challenge.

For simplicity, let $D = D_\eML{\veps ^ 2}$. For any candidate set $I$, let $P(I)$ denote its probability mass in the distribution $D$.
Define
\[
\tau \defeq \inf\set{t \geq 0: P(I_t^c) \leq \veps} = \min\set{t \geq 0: P(I_t^c) \leq \veps}.
\]

\paragraph{Case (1): Integral on $0 \leq t < \tau$.} When $t$ falls within the range $[0, \tau)$, we have $P(I_t^c) > \veps$. In this case, the pointwise inequality that relates $\ell(D, t)$ to the lifted version of $r(t)$ holds. The  following lemma is inspired by and generalizes its counterpart \cite[Theorem 1]{DBLP:journals/jacm/CharikarRWW24} which is used in showing that exact (i.e., $\veps = 0$) Maximal Lotteries have distortion at most $3$.
\begin{lemma}\label{lem:repapx_ML_large_case}
  For any fixed biased metric, any $\veps < 1 / \sqrt{10}$, and any $0 \leq t < \tau$, we have
  \[
  \ell(D, t) \leq \min\set{\frac12, r(t)} + 2\veps.
  \]
\end{lemma}

We need to integrate both sides of \cref{lem:repapx_ML_large_case} in the proof of \cref{thm:distortion_repapx_ML}. The following fact helps us control the extra area induced by the $2\veps$-lift.

\begin{fact}\label{fac:tau}
$\tau \leq 5R$.
\end{fact}

The proof of \cref{lem:repapx_ML_large_case,fac:tau} are deferred to \cref{subsec:omitted_proof_repapx_ML}.

\paragraph{Case (2): Integral on $t \geq \tau$.} When $t$ falls within the range $[\tau, +\infty)$, we have $P(I_t^c) \leq \veps$. In this case, the curve $r(t) + 2\veps$ may not stay above $\ell(D, t)$. Luckily, we know that the distortion of $\Supp\left(D(I_t^c)\right)$ is at most a constant (\cref{thm:const_dist_supp_ML}). Combining this fact with $P(I_t^c) \leq \veps$, we can directly bound the integral of $\ell(D, t)$ over $t \geq \tau$.

\begin{lemma}\label{lem:repapx_ML_small_case}
  For any fixed biased metric, we have
  \[
  \int_{\tau}^{\infty}\ell(D, t)\diff t \leq 4\veps R.
  \]
\end{lemma}

The proof of \cref{lem:repapx_ML_small_case} is deferred to \cref{subsec:omitted_proof_repapx_ML}. We present the proof of \cref{thm:distortion_repapx_ML} assuming \cref{lem:repapx_ML_large_case,fac:tau,lem:repapx_ML_small_case}.

\begin{proof}[Proof of \cref{thm:distortion_repapx_ML}]
  We can compute
  \begin{align*}
    L(D) &= \int_0^\tau \ell(D, t) \diff t + \int_\tau^\infty \ell(D, t) \diff t\\
    &\leq \int_0^\tau (r(t) + 2\veps)\diff t + \int_\tau^\infty \ell(D, t)\diff t\tag{\cref{lem:repapx_ML_large_case}}\\
    &\leq \int_0^\tau r(t)\diff t +  2 \veps \tau + 4\veps R\tag{\cref{lem:repapx_ML_small_case}}\\
    &\leq R + 14\veps R\tag{\cref{fac:tau}}.
  \end{align*}
  Applying \cref{thm:biased_metric} concludes our proof.
\end{proof}

\subsection{Omitted Proofs}\label{subsec:omitted_proof_repapx_ML}

\begin{proof}[Proof of \cref{lem:repapx_ML_large_case}]
If $P(I_t) = 0$, then $D(I_t)$ is not well-defined. We need to analyze this case separately.

Using the fact that $i^* \in I_t$, we have
\begin{align*}
\ell(D, t) = \sum_{j \not\in I_t}s_{I_t \succ j}p_j \mkern-3mu \stackrel{ \text{\cref{fac:set2dis}} }{\leq} \mkern-3mu \sum_{j \not\in I_t}s_{i^* \succ j}p_j = \sum_{j \in C}s_{i^* \succ j}p_j = s_{i^* \succ D}.
\end{align*}
Applying \cref{def:repapx_ML}, we get $\ell(D, t) \leq \frac12 + \veps ^ 2$. Then $\ell(D, t) \leq \frac12 + 2\veps$ follows immediately.

On the other hand, since we showed in the proof of \cref{cor:biased_metric} that $s_{\forall i > j, x_i - x_j \leq t} \leq s_{i^* \succ I_t^c}$, we have
\begin{align*}
  r(t) = 1 - s_{\forall i > j, x_i - x_j \leq t} \geq 1 - s_{i^* \succ I_t^c}\mkern-5mu \stackrel{ \text{\cref{fac:set2dis}} }{\geq} \mkern-5mu 1 - s_{i^* \succ D(I_t^c)} = 1 - s_{i^* \succ D}.
\end{align*} Again, applying \cref{def:repapx_ML} we get $r(t) \geq \frac12 - \veps ^ 2$.

Note that $\left(\frac12 + \veps ^ 2\right)/\left(\frac12 - \veps ^ 2\right) \leq 1 + 5 \veps ^ 2 \leq 1 + 2\veps$ if $\veps ^ 2 < 1/10$. Therefore,
\[
\ell(D, t) \leq \frac{\left(\frac12 + \veps ^ 2\right)}{\left(\frac12 - \veps ^ 2\right)} \cdot \left(\frac12 - \veps ^ 2\right) \leq (1 + 2\veps)r(t) \leq r(t) + 2\veps.
\]
The desired bound holds when $P(I_t) = 0$.

From now on, we assume $P(I_t) > 0$, so $D(I_t)$ is well-defined. Moreover, since $t < \tau$, we have $P(I_t^c) > \veps$, and hence $D(I_t^c)$ is also well-defined. We can get
  \begin{align*}
  \ell(D, t) = \sum_{j \not\in I_t}s_{I_t\succ j}p_j \mkern-3mu \stackrel{ \text{\cref{fac:set2dis}} }{\leq} \mkern-3mu \sum_{j \not\in I_t}s_{D(I_t)\succ j}p_j = s_{D(I_t)\succ D(I_t^c)} \cdot P(I_t^c),
  \end{align*} and
  \begin{align*}
  r(t) = 1 - s_{\forall i > j, x_i - x_j \leq t} \geq 1 - s_{i^* \succ I_t^c} \mkern-5mu \stackrel{ \text{\cref{fac:set2dis}} }{\geq} \mkern-5mu  1 - s_{i^* \succ D(I_t^c)}.
  \end{align*}

  The remainder of this proof aims at establishing the following inequalities, which suffice to conclude our proof.
  \[
  s_{D(I_t) \succ D(I_t^c)} \cdot P(I_t^c)\leq P(I_t^c) \cdot \left(\frac12 + \veps\right) \leq 1 - s_{i^*\succ D(I_t^c)} + 2\veps.
  \]

  \paragraph{Proving $s_{D(I_t) \succ D(I_t^c)} \cdot P(I_t^c)\leq P(I_t^c) \cdot \left(\frac12 + \veps\right)$.}

  It suffices to show $s_{D(I_t) \succ D(I_t^c)} \leq \frac12 + \veps$. We can compute
  \begin{align*}
    \frac{1}{2} + \veps ^ 2 &\geq s_{D(I_t) \succ D} \tag{\cref{def:repapx_ML}}\\
    &= P(I_t^c) \cdot s_{D(I_t)\succ D(I_t^c)} + \bigl(1-P(I_t^c)\bigr) \cdot s_{D(I_t)\succ D(I_t)}\tag{law of total probability}\\
    &= P(I_t^c) s_{D(I_t)\succ D(I_t^c)} + (1 - P(I_t^c)) \cdot \frac{1}{2} \tag{symmetry}
  \end{align*} Rearranging yields
  \begin{equation}    
  s_{D(I_t) \succ D(I_t^c)} \leq \frac12 + \frac{\veps^2}{P(I_t^c)} \leq \frac12 + \veps. \label{eq:1}
  \end{equation}

  \paragraph{Proving $P(I_t^c) \cdot\left(\frac12 + \veps\right) \leq 1 - s_{i^*\succ D(I_t^c)} + 2\veps$.} Consider the three strategies $D(I_t)$, $D(I_t^c)$ and deterministically selecting candidate $i^*$. We can compute
\begin{align*}
  \frac12 - \veps^2 &\leq s_{D \succ i^*} \tag{\cref{def:repapx_ML}}\\
&= P(I_t^c) \cdot s_{D(I_t^c)\succ i^*} + \bigl(1 - P(I_t^c)\bigr) \cdot s_{D(I_t)\succ i^*}\tag{law of total probability}\\
&\leq P(I_t^c) \cdot s_{D(I_t^c)\succ i^*} + \bigl(1 - P(I_t^c)\bigr) \cdot \bigl(s_{D(I_t)\succ D(I_t^c)} + s_{D(I_t^c)\succ i^*}\bigr)\tag{\cref{lem:triangle}}\\
&= s_{D(I_t^c)\succ i^*} + \bigl(1 - P(I_t^c)\bigr) \cdot s_{D(I_t)\succ D(I_t^c)}\\
&\leq s_{D(I_t^c)\succ i^*} + \bigl(1 - P(I_t^c)\bigr) \cdot \left(\frac{1}{2}+ \veps\right). \tag{\cref{eq:1}}
\end{align*} Rearranging yields
    \[
    P(I_t^c) \cdot \left(\frac12 + \veps\right) \leq 1 - s_{i^* \succ D(I_t^c)} + \veps ^ 2 + \veps \leq 1 - s_{i^* \succ D(I_t^c)} + 2\veps. \qedhere
    \]
\end{proof}

\begin{proof}[Proof of \cref{fac:tau}]
  By definition, $\tau$ is the smallest $t \geq 0$ that satisfies $P(I_t^c) \leq \veps$, and therefore there exists a candidate $j$ with $x_j = \tau$ and $p_j > 0$. The proof is concluded by computing
  \begin{align*}
  x_j &= \frac1n \sum_{v \in V} d(i^*, j)\\
  &\leq \frac1n \sum_{v \in V}(d(v, i^*) +  d(v, j))\tag{triangle inequality}\\
  &= \SC(j) + \SC(i^*)\\
  &\leq (4 + \sqrt{17}) \SC(i^*) + \SC(i^*) \tag{apply \cref{thm:const_dist_supp_ML} on $j \in \Supp(D) \subseteq \Supp(D_\ML)$}\\
  &\leq 5R.\tag{$R = 2\SC(i^*)$}
  \end{align*}
Thus $\tau = x_j \leq 5R$.
\end{proof}

\begin{proof}[Proof of \cref{lem:repapx_ML_small_case}]
If a candidate $j$ satisfies $j \notin I_t$ and $p_j > 0$, then $j \in I_t^c \cap \Supp(D)$, and hence $j \in \Supp(D(I_t^c))$. For any other candidate $j$, we have $s_{I_t \succ j}p_j = 0$.

We can then compute
  \begin{align*}
  \int_{\tau}^\infty \ell(D, t) \diff t&= \int_\tau^\infty\sum_{j\in \Supp(D(I_t^c))}s_{I_t \succ j}p_j\diff t\\
  &= \int_\tau^\infty\sum_{j\in \Supp(D(I_\tau^c))}s_{I_t \succ j}p_j\diff t \tag{$\Supp(D(I_t^c)) \subseteq \Supp(D(I_\tau^c))$ for all $t \geq \tau$}\\
  &= \int_\tau^\infty\sum_{j \in \Supp(D(I_\tau^c))}p_j\Pr{v \sim V}{d(j, v) - d(i^*, v) > t}\diff t \tag{\cref{fac:d_to_s_1}}\\
  &= \sum_{j \in \Supp(D(I_\tau^c))}p_j\int_\tau^\infty \Pr{v \sim V}{d(j, v) - d(i^*, v) > t}\diff t\\
  &= \sum_{j \in \Supp(D(I_\tau^c))}p_j\E{v \sim V}{\max\set{d(j, v) - d(i^*, v) - \tau, 0}}\\
  &\leq \sum_{j \in \Supp(D(I_\tau^c))}p_j\E{v \sim V}{d(j, v) - d(i^*, v)} \tag{$d(j, v) - d(i^*, v) \geq 0$ for all voter $v \in V$}\\
  &= \sum_{j \in \Supp(D(I_\tau^c))}p_j\cdot (\SC(j) - \SC(i^*))\\
  &\leq (3 + \sqrt{17})\SC(i^*)\cdot P(I_\tau^c) \tag{apply \cref{thm:const_dist_supp_ML} on $j \in \Supp(D(I_\tau^c)) \subseteq \Supp(D_\ML)$}\\
  &\leq 4\veps R. \tag{$R = 2\SC(i^*)$ and $P(I_\tau^c) \leq \veps$}
  \end{align*}
This completes the proof.
\end{proof}
 
\section{RepApx Stable Lotteries as a Voting Rule}
\cite{DBLP:conf/sigecom/CharikarRTW25} proposed \emph{Pruned Double Lotteries}, which achieve metric distortion better than $3$ by a constant. Their rule is a mixture of two components, and one of these components relies crucially on Stable Lotteries. In this section, we focus on this component and study its RepApx variant. We call the variant \nameref{box:repapx_pruned_lotteries}.

We begin by introducing \emph{quasi-kernels}.

\begin{definition}[Quasi-Kernel \cite{MISC:conf/chvatal1974directed}]
In a directed graph, a \emph{quasi-kernel} is an independent set\footnote{An independent set in a directed graph is a subset $S$ of the vertices such that no two vertices in $S$ are joined by an edge.} $S$ that satisfies the following condition: for every vertex $v \notin S$, there is a vertex $u \in S$, such that there is a directed path from $u$ to $v$ with at most $2$ edges.
\end{definition}

Every directed graph admits a quasi-kernel \cite{MISC:conf/chvatal1974directed} that can be found in linear time \cite{DBLP:journals/ipl/Croitoru15}.

Given an election instance, we define the \nameref{box:quasi_pruning} procedure \cite{DBLP:journals/jacm/CharikarRWW24,DBLP:conf/sigecom/CharikarRTW25}, which reduces the set of candidates to a quasi-kernel.

\begin{mybox}[label={box:quasi_pruning},nameref={Quasi-Kernel Pruning}]{$\theta$-Quasi-Kernel Pruning ($\theta \in \left(\frac12, 1\right]$)}
  \begin{itemize}
    \item Construct a directed graph as follows. The vertices are the candidates $C$. For all candidates $a \neq b$ with $s_{a \succ b} \geq \theta$, add a directed edge from $a$ to $b$.
    \item Compute a quasi-kernel of the graph and discard all candidates not in the quasi-kernel.
  \end{itemize}
\end{mybox}

We now define \nameref{box:repapx_pruned_lotteries}, whose exact (i.e.\@, $\veps = 0$) counterpart is implicit in \cite{DBLP:conf/sigecom/CharikarRTW25}.

\begin{mybox}[label={box:repapx_pruned_lotteries},nameref={RepApx Pruned Lotteries}]{$(\veps, k, \theta)$-RepApx Pruned Lotteries ($\veps \in (0, 1), k \in \N^+, \theta \in \left(\frac12, 1\right]$)}
  \begin{itemize}
    \item Run $\theta$-\nameref{box:quasi_pruning}.
    \item Return an $\veps$-RepApx Stable $k$-Lottery over the remaining candidates.
  \end{itemize}
\end{mybox}

\begin{definition}[$\theta$-Regularity; see, e.g.\@, \cite{DBLP:conf/sigecom/CharikarRTW25}] A preference profile is $\theta$-regular if $s_{i \succ j} < \theta$ for all $i, j \in C$.
\end{definition}

The preference profile induced by the remaining candidates after $\theta$-\nameref{box:quasi_pruning} is $\theta$-regular since the remaining candidates form an independent set in the constructed directed graph. The following lemma analyzes the metric distortion of $\veps$-RepApx Stable $k$-Lotteries relative to the remaining candidates (which form a $\theta$-regular profile). That is, the metric distortion is at most $\alpha$ if the expected social cost of the selected candidate is at most $\alpha$ times the social cost of selecting any remaining candidate.

\begin{lemma}\label[lemma]{lem:sl_quasi_kernel} Given a $\theta$-regular profile, if $\frac1{k + 1} + \veps \leq \frac2k$ and integer $k \geq 7$, an $\veps$-RepApx Stable $k$-Lottery over the candidates has metric distortion at most $1+2\lambda$, provided that $\lambda$ satisfies the constraints
 \[
  p \leq \max\set{\frac{\lambda}{\theta}(1-\theta),\frac{\lambda}{\theta}\left(1-\left(\frac{1}{k+1}+\veps\right)p^{-k}\right)}, \qquad \forall p\in [0,1].
  \] In particular, the smallest $\lambda$ that satisfies these constraints is given by the formula
\[
\lambda(\theta, k, \veps)=\frac{\theta}{1-\theta}\left(\frac{1}{\theta (k+1)}+\frac{\veps}{\theta}\right)^{1/k}.
\]
\end{lemma} The proof of \cref{lem:sl_quasi_kernel} is deferred to \cref{subsec:sl_quasi_remaining}.

We now show that the metric distortion of $(\veps, k, \theta)$-\nameref{box:repapx_pruned_lotteries} is at most a constant that depends solely on $\veps, k$, and $\theta$.

\begin{lemma}\label{lem:sl_quasi_pruning_distortion}
  The metric distortion of $(\veps, k, \theta)$-\nameref{box:repapx_pruned_lotteries} is at most
  \[
  \left(1+2\lambda(\theta, k, \veps)\right)\left(\frac{4}{\theta} - 3\right),
  \] if  $\frac1{k + 1} + \veps \leq \frac2k$ and integer $k \geq 7$.
\end{lemma}

As $\theta \to {\frac12}^+$, $k \to \infty$, and $\veps k \to 0^+$, the above expression converges to $15$.

\begin{proof}[Proof of \cref{lem:sl_quasi_pruning_distortion} Assuming \cref{lem:sl_quasi_kernel}]
Let $S$ denote the set of remaining candidates after $\theta$-\nameref{box:quasi_pruning}. Since $S$ is a quasi-kernel, at least one of the following three cases must hold: (1) $i^* \in S$, (2) there exists a candidate $j \in S$ such that $s_{j \succ i^*} \geq \theta$, in which case $\SC(j) \leq \left(\frac{2}{\theta} - 1\right) \cdot \SC(i^*)$ by \cref{lem:one-hop}, (3) there exists a candidate $j \in S$ and a candidate $k \notin S$ such that $s_{k \succ i^*} \geq \theta$ and $s_{j \succ k} \geq \theta$, in which case $\SC(j) \leq \left(\frac{4}{\theta} - 3\right) \cdot \SC(i^*)$ by \cref{lem:two-hop-var}. Among these, the third case yields the worst distortion bound, and therefore it suffices to analyze this case.

Let $j^*$ denote a candidate in $S$ with the minimum social cost. Then, by $\theta$-regularity, we can apply \cref{lem:sl_quasi_kernel} to show that the expected social cost of an $\veps$-RepApx Stable $k$-Lottery over $S$ is at most
\[
  \left(1+2\lambda(\theta, k, \veps)\right)\cdot\SC(j^*) \leq \left(1+2\lambda(\theta, k, \veps)\right)\cdot \SC(j) \leq \left(1+2\lambda(\theta, k, \veps)\right)\left(\frac{4}{\theta} - 3\right)\cdot \SC(i^*).
\] This completes the proof.
\end{proof}

\subsection{Metric Distortion of RepApx Stable Lotteries Among Regular Profiles}\label{subsec:sl_quasi_remaining}
We begin with a lemma capturing the intuition that if many voters prefer a single candidate $i \in C \setminus J$ to every candidate in $J$, then a RepApx Stable Lottery cannot assign too much probability mass to $J$. The proof idea of our lemma follows closely from that of \cite[Lemma 5.16]{DBLP:conf/sigecom/CharikarRTW25}, the \emph{exact Stable Lotteries} counterpart of our lemma.

\begin{lemma}\label[lemma]{lem:stability_representative}
  Fix an $\veps$-RepApx Stable Lottery $D_\ekSL{\veps}{k}$. Consider any candidate subset $J \subsetneq C$ and any candidate $i \in C \setminus J$. Let $p_J$ denote the probability that $D_\ekSL{\veps}{k}$ selects a candidate in $J$. If $p_J > 0$, then
  \[
  s_{i \succ J} \leq \left(\frac{1}{k + 1} + \veps\right)p_J^{-k}.
  \]
\end{lemma}

\begin{proof}[Proof of \cref{lem:stability_representative}]
 Let $D = D_\ekSL{\veps}{k}$. Let $S$ be the set of voters who prefer $i$ to all candidates in $J$. Then we have
\begin{align*}
\frac1{k + 1} + \veps &\geq \Pr{v\sim V}{i\succ_v D^k} \tag{\cref{def:repapx_SL}}\\
&= s_{i\succ J}\cdot\Pr{v\sim S}{i\succ_v D^k} + (1-s_{i\succ J})\cdot\Pr{v\sim C\setminus S}{i\succ_v D^k} \tag{law of total probability}\\
&\geq s_{i\succ J}\cdot\Pr{v\sim S}{i\succ_v D^k}\\
&\geq s_{i\succ J}\cdot p^k_J. \tag{$i\succ_v J$ for all voter $v \in S$}
\end{align*}
Rearranging yields the claim.
\end{proof}

We now establish the metric distortion of RepApx Stable Lotteries among regular profiles.

\begin{proof}[Proof of \cref{lem:sl_quasi_kernel}]
By \cref{cor:biased_metric}, it suffices to verify the biased-metric condition \[\sum_{j\in J} s_{i ^ *\succ j}\,p_j \leq \lambda(1-s_{i ^ *\succ J}) \qquad \text{for all $J \subseteq C\setminus\set{i^*}$}.\] Writing $p_J=\sum_{j\in J}p_j$ and using $\theta$-regularity, this condition is implied by
\[
 p_J \leq \frac{\lambda}{\theta}\cdot(1-s_{i ^*\succ J}).
\]

We derive two sufficient cases.

\begin{itemize}
  \item If $p_J\leq \frac{\lambda}{\theta}(1-\theta)$, the inequality holds immediately since \[s_{i^* \succ J} \mkern-3mu \stackrel{ \text{\cref{fac:set2dis}} }{\leq} \mkern-3mu \min_{j \in J}s_{i^* \succ j} \mkern-3mu \stackrel{ \text{$\theta$-regularity}}{<} \mkern-3mu \theta.\]
  \item Otherwise, \cref{lem:stability_representative} gives $s_{i^* \succ J} \leq \left(\frac1{k + 1} + \veps\right)p_J^{-k}$, equivalent to $\frac{\lambda}{\theta}\cdot\left(1-\left(\frac{1}{k+1}+\veps\right)p_J^{-k}\right)\leq\frac{\lambda}{\theta}\cdot(1-s_{i^*\succ J})$. Hence,
\[
p_J\leq\frac{\lambda}{\theta}\cdot\left(1-\left(\frac{1}{k+1}+\veps\right)p_J^{-k}\right)
\]
is also a sufficient condition for $p_J \leq \frac{\lambda}{\theta}\cdot(1-s_{i ^*\succ J})$.
\end{itemize}
 
To complete the proof, we choose $\lambda$ so that for every $p_J \in [0, 1]$, at least one of the above two cases applies. The stated expression for $\lambda$ is the minimum value that satisfies this requirement; the derivation is deferred to \cref{sec:missing_sl_quasi_kernel}.
\end{proof}
 
\section{Mixing Rules}

Neither RepApx Maximal Lotteries nor \nameref{box:repapx_pruned_lotteries} achieve distortion below $3$ on their own. To combine the strength of these two rules, we introduce the notion of a \emph{strongly consistent} biased metric, proposed by \cite{DBLP:journals/jacm/CharikarRWW24}.\footnote{They defined $(\alpha, \beta)$-consistency in the main body and a slightly stronger variant in the appendix; we adopt the latter here and name it ``strong $(\alpha, \beta)$-consistency.''} By analyzing the behavior of each rule on biased metrics that are, or are not, strongly consistent, we show that the two rules are complementary in the sense that an appropriate mixture of them achieves distortion strictly less than $3$.

\begin{definition}(Strongly Consistent Biased Metric)
\label{def:consistent}
A biased metric for an election instance is \emph{strongly $(\alpha,\beta)$-consistent} if whenever $s_{a \succ b} \geq \beta$, we have $x_a - x_b \leq \alpha R$.
\end{definition}

This definition quantifies how well the biased metric aligns with the preference profile. Intuitively, if candidate $a$ defeats candidate $b$ by a non-negligible margin, then $a$ should not be significantly farther from the optimal candidate $i^*$ than $b$. This consistency condition can also be directly characterized in terms of the function $r$, as formalized by the following lemma.

\begin{lemma}[similar to {\cite[Proposition 3]{DBLP:journals/jacm/CharikarRWW24}}]\label[lemma]{lem:r_consist}
  For any fixed biased metric, if $r(\alpha R) < \beta$, then the metric is strongly $(\alpha, \beta)$-consistent.
\end{lemma}
\begin{proof}[Proof of \cref{lem:r_consist}]
  By definition, $r(\alpha R) < \beta$ implies $s_{\forall i \succ j, x_i - x_j \leq \alpha R} > 1 - \beta$. Now fix any pair of candidates $a, b$ with $s_{a \succ b} > \beta$. We then have
  \[
s_{\forall i \succ j, x_i - x_j \leq \alpha R} + s_{a \succ b} > 1 - \beta + \beta = 1.
  \] Therefore, there exists at least one voter $v$ who belongs to the intersection \[S_{\forall i \succ j, x_i - x_j \leq \alpha R} \cap S_{a \succ b}.\] For this voter $v$, since $a \succ_v b$ and $v \in S_{\forall i \succ j, x_i - x_j \leq \alpha R}$, it follows in particular that $x_a - x_b \leq \alpha R$. Hence, the metric is strongly $(\alpha, \beta)$-consistent.
\end{proof}

\subsection{Part I: Inconsistent Biased Metrics}
In this part, we analyze metric distortion under inconsistent biased metrics. Specifically, we consider the case in which the biased metric is \emph{not} strongly $\left(\alpha, \frac12 + \tilde{\beta}\right)$-consistent. In this case, RepApx Maximal Lotteries attain distortion below $3$ by a constant with appropriate parameter choices. As for the other component, although \nameref{box:repapx_pruned_lotteries} do not benefit from this regime, their metric distortion is at most a constant for all metrics (\cref{lem:sl_quasi_pruning_distortion}).

\begin{lemma}\label[lemma]{lem:approx_ML_inconsistent}
  Given a biased metric that is not strongly $\left(\alpha, \frac12 + \tilde{\beta}\right)$-consistent, if $\veps < 1 / \sqrt{10}$, then an $\veps^2$-RepApx Maximal Lottery $D_\eML{\veps ^ 2}$ has metric distortion at most $3 + 28\veps - 2\alpha\tilde{\beta}$.
\end{lemma}

In particular, when $\veps$ is chosen sufficiently small relative to $\alpha\tilde{\beta}$, the resulting distortion is strictly less than $3$.

\begin{proof}[Proof of \cref{lem:approx_ML_inconsistent}]
By \cref{lem:r_consist}, in this inconsistent regime we have $r(\alpha R) \geq \frac12 + \tilde{\beta}$.

  Let $D = D_\eML{\veps ^ 2}$. By \cref{lem:repapx_ML_large_case}, for all $t \leq \tau$, it holds that
  \[
  \ell(D, t) \leq \min\set{\frac12, r(t)} + 2\veps.
  \]
 Therefore, $\ell(D, t)$ lies below the line $\frac12 + 2\veps$ for all $t \geq 0$ (since $\ell(D, t)$ is monotonically non-increasing in $t$), and below $r(t) + 2\veps$ when $t \leq \tau$. Moreover, for all $t < \alpha R$,
 \begin{equation}\label{eq:2}
 \ell(D, t) \leq \frac12 + 2\veps \leq r(\alpha R) + 2\veps - \tilde{\beta} \leq r(t) + 2\veps - \tilde{\beta}.
 \end{equation}
 Geometrically, this creates a rectangular region of area at least $\alpha \tilde{\beta} R$ that lies above $\ell(D, t)$ and below $r(t) + 2\veps$. When this gap is large, we obtain a sharper bound on $L(D) / R$.
  
We distinguish two cases, depending on the relative positions of $\alpha R$ and $\tau$.

\begin{figure}[t!]
    \centering
    \begin{subfigure}[t]{0.49\textwidth}
        \centering
        \begin{tikzpicture}[
    scale=0.53,
myredline/.style={color=red!80!black, line width=1pt},
    myredfill/.style={color=red!60, opacity=0.8},
    myblueline/.style={color=blue!40!teal!80!black, line width=1pt},
    mybluefill/.style={color=blue!40!teal!30, opacity=0.8},
    axis/.style={line width=1pt, black},
    dashed guide/.style={dashed, thin, black!80, dash pattern=on 5pt off 3pt},
dot/.style={circle, fill=white, draw=#1, line width=1pt, inner sep=1pt},
filldot/.style={circle, fill=#1, inner sep=1pt}
]

\def\yTop{7}
    \def\yB{4.5}
    \def\yHalf{3.5}
    \def\yZero{0}
    \def\xMax{9}
    \def\xTau{4.0}
    \def\xAlphaR{5.0}
    \def\xMu{5.9}
    \def\yEps{0.2}

\fill[mybluefill] 
        (0,\yZero) -- (0,3.2) -- (2.2,3.2) -- (2.2,2.8) -- (4.0,2.8) -- 
        (4.0,2.4) -- (6.3,2.4) -- (6.3,1.9) -- (8,1.9) -- (8,\yZero) -- cycle;

    \begin{scope}
        \clip (0,\yHalf) rectangle (\xAlphaR, \yB);
        \fill[myredfill] 
            (0,0) -- (0,6.2) -- (1.8,6.2) -- (1.8,5.5) -- (3.5, 5.5) -- 
            (3.5,4.8) -- (5.5,4.8) -- (5.5,1.5) -- (7.0, 1.5) -- 
            (7.0,1.0) -- (8.5,1.0) -- (8.5,\yEps) -- (\xMax,\yEps) -- (\xMax, 0) -- cycle;
    \end{scope}

\draw[dashed guide, very thick] (0, \yHalf) -- (\xMax, \yHalf) node[left, pos=0, xshift=-0.1cm] {$\frac{1}{2} + 2\veps$};
    \draw[axis] (0, \yTop) -- (0, 0) -- (\xMax, 0) node[below] {$t$};
    \draw[dashed guide, very thick] (\xTau, \yTop) -- (\xTau, 0) node[below] {$\tau$};
    \draw[dashed guide, color=red!60, very thick] (\xAlphaR, \yHalf) -- (\xAlphaR, 0) node[below] {$\alpha R$};

    \draw[line width=1pt] (0, \yTop - 0.03) -- (-0.2, \yTop - 0.03) node[left] {$1 + 2\veps$};
    \draw[line width=1pt] (0, \yB - 0.03) -- (-0.2, \yB - 0.03) node[left] {$\frac{1}{2} + 2\veps + \tilde{\beta}$};
    \draw[line width=1pt] (0, \yEps - 0.03) -- (-0.2, \yEps - 0.03) node[left] {$2\veps$};

\coordinate (B_start) at (0, 3.2);
    \coordinate (B1_R) at (2.2, 3.2); \coordinate (B1_low) at (2.2, 2.8);
    \coordinate (B2_R) at (4.0, 2.8); \coordinate (B2_low) at (4.0, 2.4);
    \coordinate (B3_R) at (6.3, 2.4); \coordinate (B3_low) at (6.3, 1.9);
    \coordinate (B4_R) at (8, 1.9); \coordinate (B4_low) at (8, 0);
    \coordinate (B_end) at (\xMax, 0);

    \draw[myblueline] (B_start) -- (B1_R);
    \draw[myblueline] (B1_low) -- (B2_R);
    \draw[myblueline] (B2_low) -- (B3_R);
    \draw[myblueline] (B3_low) -- (B4_R);
    \draw[myblueline] (B4_low) -- (B_end);

    \draw[myblueline, dashed] (B1_R) -- (B1_low);
    \draw[myblueline, dashed] (B2_R) -- (B2_low);
    \draw[myblueline, dashed] (B3_R) -- (B3_low);
    \draw[myblueline, dashed] (B4_R) -- (B4_low);

\def\blueColor{blue!40!teal!80!black}
    
    \node[dot=\blueColor] at (B1_R) {};
    \node[filldot=\blueColor] at (B1_low) {};
    
    \node[dot=\blueColor] at (B2_R) {};
    \node[filldot=\blueColor] at (B2_low) {};
    
    \node[dot=\blueColor] at (B3_R) {};
    \node[filldot=\blueColor] at (B3_low) {};
    
    \node[dot=\blueColor] at (B4_R) {};
    \node[filldot=\blueColor] at (B4_low) {};

\coordinate (R_start) at (0, 6.2);
    \coordinate (R1_R) at (1.8, 6.2); \coordinate (R1_low) at (1.8, 5.5);
    \coordinate (R2_R) at (3.5, 5.5); \coordinate (R2_low) at (3.5, 4.8);
    \coordinate (R3_R) at (5.5, 4.8); \coordinate (R3_low) at (5.5, 1.5);
    \coordinate (R4_R) at (7.0, 1.5); \coordinate (R4_low) at (7.0, 1.0);
    \coordinate (R5_R) at (8.5, 1.0); \coordinate (R5_low) at (8.5, \yEps);
    \coordinate (R6_R) at (\xMax, \yEps); \coordinate (R_end) at (\xMax, 0);

    \draw[myredline] (R_start) -- (R1_R);
    \draw[myredline] (R1_low) -- (R2_R);
    \draw[myredline] (R2_low) -- (R3_R);
    \draw[myredline] (R3_low) -- (R4_R);
    \draw[myredline] (R4_low) -- (R5_R);
    \draw[myredline] (R5_low) -- (R6_R);

    \draw[myredline, dashed] (R1_R) -- (R1_low);
    \draw[myredline, dashed] (R2_R) -- (R2_low);
    \draw[myredline, dashed] (R3_R) -- (R3_low);
    \draw[myredline, dashed] (R4_R) -- (R4_low);
    \draw[myredline, dashed] (R5_R) -- (R5_low);

\def\redColor{red!80!black}

    \node[dot=\redColor] at (R1_R) {};
    \node[filldot=\redColor] at (R1_low) {};
    \node[dot=\redColor] at (R2_R) {};
    \node[filldot=\redColor] at (R2_low) {};
    \node[dot=\redColor] at (R3_R) {};
    \node[filldot=\redColor] at (R3_low) {};
    \node[dot=\redColor] at (R4_R) {};
    \node[filldot=\redColor] at (R4_low) {};
    \node[dot=\redColor] at (R5_R) {};
    \node[filldot=\redColor] at (R5_low) {};

\node[blue!40!teal!80!black, anchor=west] at (7.5, 2.5) {$\ell(D, t)$};
    
\node[red!80!black, anchor=south west] at (5.5, 4.5) {$r(t) + 2\veps$};

\end{tikzpicture}
         \caption{$\ell(D, t)$ and $r(t) + 2\veps$ if $\alpha R \geq \tau$}
        \label{fig:repapx_ML_ab_fig1}
    \end{subfigure}~ 
    \begin{subfigure}[t]{0.49\textwidth}
        \centering
        \begin{tikzpicture}[
    scale=0.53,
myredline/.style={color=red!80!black, line width=1pt},
    myredfill/.style={color=red!60, opacity=0.8},
    myblueline/.style={color=blue!40!teal!80!black, line width=1pt},
    mybluefill/.style={color=blue!40!teal!30, opacity=0.8},
    axis/.style={line width=1pt, black},
    dashed guide/.style={dashed, thin, black!80, dash pattern=on 5pt off 3pt},
dot/.style={circle, fill=white, draw=#1, line width=1pt, inner sep=1pt},
filldot/.style={circle, fill=#1, inner sep=1pt}
]

\def\yTop{7}
    \def\yB{4.5}
    \def\yHalf{3.5}
    \def\yZero{0}
    \def\xMax{9}
    \def\xTau{4.0}
    \def\xAlphaR{3.0}
    \def\xMu{5.9}
    \def\yEps{0.2}

\fill[mybluefill] 
        (0,\yZero) -- (0,3.2) -- (2.2,3.2) -- (2.2,2.8) -- (4.0,2.8) -- 
        (4.0,2.4) -- (6.3,2.4) -- (6.3,1.9) -- (8,1.9) -- (8,\yZero) -- cycle;

    \begin{scope}
        \clip (0,\yHalf) rectangle (\xAlphaR, \yB);
        \fill[myredfill] 
            (0,0) -- (0,6.2) -- (1.8,6.2) -- (1.8,5.5) -- (3.5, 5.5) -- 
            (3.5,4.8) -- (5.5,4.8) -- (5.5,1.5) -- (7.0, 1.5) -- 
            (7.0,1.0) -- (8.5,1.0) -- (8.5,\yEps) -- (\xMax,\yEps) -- (\xMax, 0) -- cycle;
    \end{scope}

\draw[dashed guide, very thick] (0, \yHalf) -- (\xMax, \yHalf) node[left, pos=0, xshift=-0.1cm] {$\frac{1}{2} + 2\veps$};
    \draw[axis] (0, \yTop) -- (0, 0) -- (\xMax, 0) node[below] {$t$};
    \draw[dashed guide, very thick] (\xTau, \yTop) -- (\xTau, 0) node[below] {$\tau$};
    \draw[dashed guide, color=red!60, very thick] (\xAlphaR, \yHalf) -- (\xAlphaR, 0) node[below] {$\alpha R$};

    \draw[line width=1pt] (0, \yTop - 0.03) -- (-0.2, \yTop - 0.03) node[left] {$1 + 2\veps$};
    \draw[line width=1pt] (0, \yB - 0.03) -- (-0.2, \yB - 0.03) node[left] {$\frac{1}{2} + 2\veps + \tilde{\beta}$};
    \draw[line width=1pt] (0, \yEps - 0.03) -- (-0.2, \yEps - 0.03) node[left] {$2\veps$};

\coordinate (B_start) at (0, 3.2);
    \coordinate (B1_R) at (2.2, 3.2); \coordinate (B1_low) at (2.2, 2.8);
    \coordinate (B2_R) at (4.0, 2.8); \coordinate (B2_low) at (4.0, 2.4);
    \coordinate (B3_R) at (6.3, 2.4); \coordinate (B3_low) at (6.3, 1.9);
    \coordinate (B4_R) at (8, 1.9); \coordinate (B4_low) at (8, 0);
    \coordinate (B_end) at (\xMax, 0);

    \draw[myblueline] (B_start) -- (B1_R);
    \draw[myblueline] (B1_low) -- (B2_R);
    \draw[myblueline] (B2_low) -- (B3_R);
    \draw[myblueline] (B3_low) -- (B4_R);
    \draw[myblueline] (B4_low) -- (B_end);

    \draw[myblueline, dashed] (B1_R) -- (B1_low);
    \draw[myblueline, dashed] (B2_R) -- (B2_low);
    \draw[myblueline, dashed] (B3_R) -- (B3_low);
    \draw[myblueline, dashed] (B4_R) -- (B4_low);

\def\blueColor{blue!40!teal!80!black}
    
    \node[dot=\blueColor] at (B1_R) {};
    \node[filldot=\blueColor] at (B1_low) {};
    
    \node[dot=\blueColor] at (B2_R) {};
    \node[filldot=\blueColor] at (B2_low) {};
    
    \node[dot=\blueColor] at (B3_R) {};
    \node[filldot=\blueColor] at (B3_low) {};
    
    \node[dot=\blueColor] at (B4_R) {};
    \node[filldot=\blueColor] at (B4_low) {};

\coordinate (R_start) at (0, 6.2);
    \coordinate (R1_R) at (1.8, 6.2); \coordinate (R1_low) at (1.8, 5.5);
    \coordinate (R2_R) at (3.5, 5.5); \coordinate (R2_low) at (3.5, 4.8);
    \coordinate (R3_R) at (5.5, 4.8); \coordinate (R3_low) at (5.5, 1.5);
    \coordinate (R4_R) at (7.0, 1.5); \coordinate (R4_low) at (7.0, 1.0);
    \coordinate (R5_R) at (8.5, 1.0); \coordinate (R5_low) at (8.5, \yEps);
    \coordinate (R6_R) at (\xMax, \yEps); \coordinate (R_end) at (\xMax, 0);

    \draw[myredline] (R_start) -- (R1_R);
    \draw[myredline] (R1_low) -- (R2_R);
    \draw[myredline] (R2_low) -- (R3_R);
    \draw[myredline] (R3_low) -- (R4_R);
    \draw[myredline] (R4_low) -- (R5_R);
    \draw[myredline] (R5_low) -- (R6_R);

    \draw[myredline, dashed] (R1_R) -- (R1_low);
    \draw[myredline, dashed] (R2_R) -- (R2_low);
    \draw[myredline, dashed] (R3_R) -- (R3_low);
    \draw[myredline, dashed] (R4_R) -- (R4_low);
    \draw[myredline, dashed] (R5_R) -- (R5_low);

\def\redColor{red!80!black}

    \node[dot=\redColor] at (R1_R) {};
    \node[filldot=\redColor] at (R1_low) {};
    \node[dot=\redColor] at (R2_R) {};
    \node[filldot=\redColor] at (R2_low) {};
    \node[dot=\redColor] at (R3_R) {};
    \node[filldot=\redColor] at (R3_low) {};
    \node[dot=\redColor] at (R4_R) {};
    \node[filldot=\redColor] at (R4_low) {};
    \node[dot=\redColor] at (R5_R) {};
    \node[filldot=\redColor] at (R5_low) {};

\node[blue!40!teal!80!black, anchor=west] at (7.5, 2.5) {$\ell(D, t)$};
    
\node[red!80!black, anchor=south west] at (5.5, 4.5) {$r(t) + 2\veps$};

\end{tikzpicture}
         \caption{$\ell(D, t)$ and $r(t) + 2\veps$ if $\alpha R < \tau$}
        \label{fig:repapx_ML_ab_fig2}
    \end{subfigure}
    \caption{Pictorial illustration of $\ell(D, t)$ and $r(t) + 2\veps$}
\end{figure}
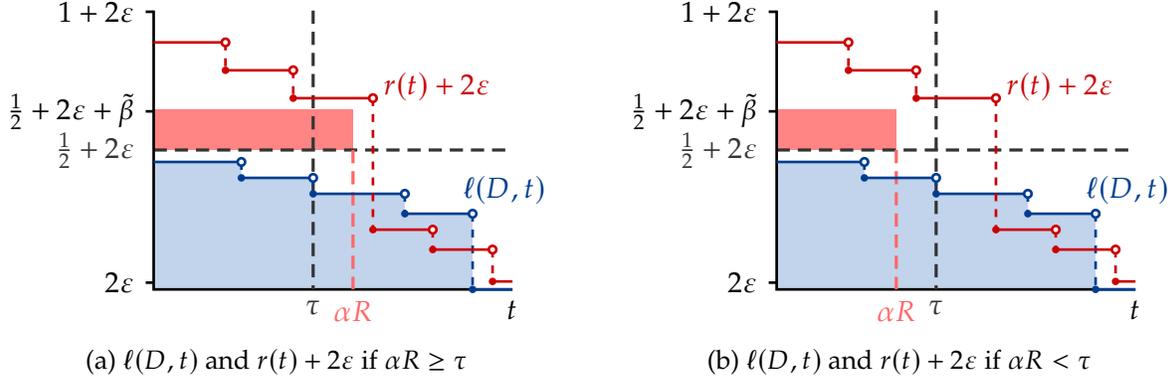

\paragraph{Case 1: $\alpha R \geq \tau$ (see \cref{fig:repapx_ML_ab_fig1}).} Split the integral defining $L(D)$ into two parts: $0 \leq t < \alpha R$ and $t \geq \alpha R$. We can compute,
  \begin{align*}
    L(D) &= \int_0^{\alpha R}\ell(D, t)\diff t + \int_{\alpha R}^\infty \ell(D, t)\diff t\\
    &\leq \int_0^{\alpha R}\ell(D, t)\diff t + \int_{\tau}^\infty \ell(D, t)\diff t \tag{$\tau \leq \alpha R$}\\
    &\leq \int_0^{\alpha R}(r(t) + 2\veps - \tilde{\beta})\diff t + 4 \veps R \tag{\cref{eq:2,lem:repapx_ML_small_case}}\\
    &\leq (1 + (4 + 2\alpha)\veps - \alpha\tilde{\beta})R.
  \end{align*} Since $\alpha R \cdot r(\alpha R) \leq R$ (as the integral region represented by $R$ includes a rectangle with sides $\alpha R$ and $r(\alpha R)$) and $r(\alpha R) \geq \frac12$, we have $\alpha \leq 2$. Hence, $L(D) \leq (1 + 8\veps - \alpha\tilde{\beta})R$.

  \paragraph{Case 2: $\alpha R < \tau$ (see \cref{fig:repapx_ML_ab_fig2}).} Split the integral defining $L(D)$ into three parts: $0 \leq t < \alpha R$, $\alpha R \leq t < \tau$, and $t \geq \tau$. We can compute
  \begin{align*}
  L(D) &= \int_0^{\alpha R}\ell(D, t)\diff t + \int_{\alpha R}^{\tau}\ell(D, t)\diff t + \int_{\tau}^\infty \ell(D, t) \diff t\\
  &\leq \int_0^{\alpha R}(r(t) + 2\veps - \tilde{\beta})\diff t + \int_{\alpha R}^\tau(r(t) + 2\veps) \diff t + 4 \veps R\tag{\cref{eq:2,lem:repapx_ML_large_case,lem:repapx_ML_small_case}}\\
  &= \int_0^\tau r(t)\diff t + 2\veps\tau + 4\veps R - \alpha\tilde{\beta}R\\
  &\leq (1 + 4\veps - \alpha\tilde{\beta})R + 2\veps\tau\\
  &\leq (1 + 14\veps - \alpha\tilde{\beta})R. \tag{\cref{fac:tau}}
  \end{align*}

  Applying \cref{thm:biased_metric} completes the proof.
\end{proof}
\subsection{Part II: Consistent Biased Metrics}

In this part, we study the regime in which the biased metric is strongly $\left(\alpha, \frac12 + \tilde{\beta}\right)$-consistent. In contrast to Part I, RepApx Maximal Lotteries may incur distortion larger than $3$, although not by too much (\cref{thm:distortion_repapx_ML}). We therefore turn to \nameref{box:repapx_pruned_lotteries} and show that, with appropriate parameter choices, they achieve distortion below $3$ by a constant.

We already know that the metric distortion of RepApx Stable Lotteries relative to the remaining candidates after \nameref{box:quasi_pruning} is low (\cref{lem:sl_quasi_kernel}). In the proof of \cref{lem:sl_quasi_pruning_distortion}, we applied \cref{lem:two-hop-var} to a length-$2$ path on a quasi-kernel to argue that restricting the choice of the winner to the remaining candidates does not significantly increase the optimal social cost. Under strong consistency, we can replace \cref{lem:two-hop-var} with the following stronger guarantee.

\begin{lemma}\label[lemma]{lem:quasi-kernel-optimal}
For any strongly $(\alpha,\beta)$-consistent biased metrics, there exists some candidate $j$ among the remaining candidates after $\beta$-\nameref{box:quasi_pruning} such that $\SC(j)\leq(1+4\alpha)\SC(i^*)$.
\end{lemma}

\begin{proof}
Let $S$ denote the set of remaining candidates after $\beta$-\nameref{box:quasi_pruning}. Since $S$ is a quasi-kernel, one of the following three cases must hold: (1) $i^* \in S$, (2) there exists a candidate $j \in S$ such that $s_{j \succ i^*} \geq \beta$, (3) there exists a candidate $j \in S$ and a candidate $k \notin S$ such that $s_{k \succ i^*} \geq \beta$ and $s_{j \succ k} \geq \beta$. Among these, the third case yields the worst distortion bound, and it suffices to analyze this case.

By strong $(\alpha, \beta)$-consistency, we have $x_k - x_{i^*} \leq \alpha R$ and $x_j - x_k \leq \alpha R$. It immediately follows that $x_j \leq 2\alpha R$.

Finally, we can compute
\begin{align*}
\SC(j) &= \E{v \sim V}{d(j, v)}\\
&\leq \E{v\sim V}{d(i^*, v) + x_j} \tag{triangle inequality}\\
&= \SC(i^*) + x_j \\
&\leq \SC(i^*) + 2\alpha R\\
&\leq (1 + 4\alpha)\SC(i^*). \tag{$R = 2\SC(i^*)$}
\end{align*} This completes the proof.
\end{proof}

Motivated by \cref{lem:quasi-kernel-optimal}, we choose the margin parameter --- the second parameter in the definition of strong consistency --- as the threshold parameter $\theta$ in $\left(\veps, k, \theta\right)$-\nameref{box:repapx_pruned_lotteries}. Accordingly, we run $\left(\veps, k, \frac12 + \tilde{\beta}\right)$-\nameref{box:repapx_pruned_lotteries}, which uses $\left(\frac12 + \tilde{\beta}\right)$-\nameref{box:quasi_pruning} as a subroutine. The metric distortion of $\left(\veps, k, \frac12 + \tilde{\beta}\right)$-\nameref{box:repapx_pruned_lotteries} under strongly $\left(\alpha, \frac12 + \tilde{\beta}\right)$-consistent biased metrics is established in the following lemma.

\begin{lemma}\label[lemma]{lem:approx_SL_quasi_consistent}
Given a strongly $\left(\alpha, \frac12 + \tilde{\beta}\right)$-consistent biased metric, the metric distortion of $\left(\veps, k, \frac12 + \tilde{\beta}\right)$-\nameref{box:repapx_pruned_lotteries} is at most
\[
\left(1+2\lambda\left(\frac12 + \tilde{\beta}, k, \veps\right)\right)(1+4\alpha),
\] if $\frac1{k + 1} + \veps \leq \frac2k$ and integer $k \geq 7$.
\end{lemma}
\begin{proof}[Proof of \cref{lem:approx_SL_quasi_consistent}]
  Let $S$ denote the set of remaining candidates after $\left(\frac12 + \tilde{\beta}\right)$-\nameref{box:quasi_pruning}. By \cref{lem:quasi-kernel-optimal}, there exists some candidate $j \in S$ such that $\SC(j) \leq (1 + 4\alpha)\SC(i^*)$.
  
  Let $j^*$ denote a candidate in $S$ with the minimum social cost. Then, by $\left(\frac12 + \tilde{\beta}\right)$-regularity, we can apply \cref{lem:sl_quasi_kernel} to show that the expected social cost of an $\veps$-RepApx Stable $k$-Lottery over $S$ is at most
  \[
\left(1 + 2\lambda\left(\frac12 + \tilde{\beta}, k, \veps\right)\right)\cdot \SC(j^*) \leq \left(1 + 2\lambda\left(\frac12 + \tilde{\beta}, k, \veps\right)\right)\cdot \SC(j) \leq \left(1+2\lambda\left(\frac12 + \tilde{\beta}, k, \veps\right)\right)(1+4\alpha) \cdot \SC(i^*).
  \]
  This completes the proof.
\end{proof}

In particular, we will see that for certain fixed values of $k$, $\veps$, $\alpha$, and $\tilde{\beta}$, the above bound for metric distortion can achieve a value strictly less than $3$. 

\subsection{Mixing Two Rules Together}
We now demonstrate how to choose parameters for RepApx Maximal Lotteries and \nameref{box:repapx_pruned_lotteries}, as well as the probability of running each of them. Our goal is for the randomization to achieve metric distortion strictly less than $3$ under both the consistent biased metrics case and the inconsistent one.

Recall that in the previous discussion, we defined the function
\[
\lambda(\theta, k, \veps) = \frac{\theta}{1-\theta}\left(\frac{1}{\theta (k+1)}+\frac{\veps}{\theta}\right)^{1/k}.
\]
With this, we summarize the bounds established so far in \cref{tab:thresholds}.

\begin{table}[t!]
\centering
\begin{tabular}{lp{8cm}}
\toprule
\textbf{Voting Rule} & \textbf{Metric Distortion} \\
\midrule
\multicolumn{2}{c}{\textit{Not Strongly $\left(\alpha, \frac12 + \tilde{\beta}\right)$-Consistent Biased Metric}} \\
\midrule
$\veps^2$-RepApx Maximal Lotteries& $3 + 28\veps - 2\alpha\tilde{\beta}$ \\
& \textcolor{gray}{\footnotesize (\cref{lem:approx_ML_inconsistent})} \\
\addlinespace
$\left(\veps, k, \frac12 + \tilde{\beta}\right)$-\nameref{box:repapx_pruned_lotteries} & $\left(1+2\lambda\left(\frac12 + \tilde{\beta}, k, \veps\right)\right)\left(\frac{4}{1/2 + \tilde{\beta}} - 3\right)$\\
& \textcolor{gray}{\footnotesize (\cref{lem:sl_quasi_pruning_distortion})} \\
\midrule
\multicolumn{2}{c}{\textit{Strongly $\left(\alpha, \frac12 + \tilde{\beta}\right)$-Consistent Biased Metric}} \\
\midrule
$\veps^2$-RepApx Maximal Lotteries& $3 + 28\veps$ \\
& \textcolor{gray}{\footnotesize (\cref{thm:distortion_repapx_ML})} \\
\addlinespace
$\left(\veps, k, \frac12 + \tilde{\beta}\right)$-\nameref{box:repapx_pruned_lotteries} & $\left(1+2\lambda\left(\frac12 + \tilde{\beta}, k, \veps\right)\right)(1+4\alpha)$ \\
& \textcolor{gray}{\footnotesize (\cref{lem:approx_SL_quasi_consistent})} \\
\bottomrule
\end{tabular}
\caption{Metric distortion Upper Bounds of $\veps ^ 2$-RepApx Maximal Lotteries and $\left(\veps, k, \frac12 + \tilde{\beta}\right)$-\nameref{box:repapx_pruned_lotteries} under consistent and inconsistent biased metrics.}
\label{tab:thresholds}
\end{table}

The following rule can be viewed as a RepApx generalization of Pruned Double Lotteries from \cite{DBLP:conf/sigecom/CharikarRTW25}.

\begin{theorem}[Mixing Theorem]\label{thm:mixing_thm}
  There exist parameters
  \[\mu, \tilde{\beta}, \veps_1, \veps_2 \in (0, 1) \qquad \text{and} \qquad k \in \N^+\]
  such that the following randomized voting rule achieves metric distortion strictly less than $3$:
  \begin{itemize}
    \item with probability $\mu$, run an $\veps_1^2$-RepApx Maximal Lottery;
    \item with probability $1 - \mu$, run an $\left(\veps_2, k, \frac12 + \tilde{\beta}\right)$-RepApx Pruned Lottery.
  \end{itemize}
\end{theorem}

\begin{proof}[Proof of \cref{thm:mixing_thm}]
If the biased metric is strongly $\left(\alpha, \frac12 + \tilde{\beta}\right)$-consistent, then by \cref{thm:distortion_repapx_ML,lem:approx_SL_quasi_consistent}, a sufficient condition for the metric distortion to be less than $3$ is
\[
\mu (3+28\veps_1) + (1 - \mu) \left(1+2\lambda\left(\frac12 + \tilde{\beta}, k, \veps_2\right)\right)(1+4\alpha)< 3.
\]

Otherwise, when the biased metric is not strongly $\left(\alpha, \frac12 + \tilde{\beta}\right)$-consistent, by \cref{lem:approx_ML_inconsistent,lem:sl_quasi_pruning_distortion}, a sufficient condition is
\[
\mu(3 + 28\veps_1 - 2\alpha\tilde{\beta}) + (1 - \mu)\left(1+2\lambda\left(\frac12 + \tilde{\beta}, k, \veps_2\right)\right)\left(\frac{4}{1/2 + \tilde{\beta}} - 3\right) < 3.
\]

It suffices to show the existence of parameters that satisfy these two conditions simultaneously.

For example, one may choose
\[\alpha = \frac1{24}\cdot\frac{\ln(k/4)}k, \tilde{\beta} = \frac19 \cdot\frac{\ln(k/4)}k, \veps_1 = \frac1{150000}\left(\frac{\ln(k/4)}k\right)^3, \veps_2 = \frac1k, \text{ and } \mu = 1 - \frac1{2000}\left(\frac{\ln(k/4)}k\right)^2\] for any integer $k \geq 7$.

 The verification of the above inequalities under this parameter choice is deferred to \cref{sec:missing_mixing_thm}.
\end{proof}

By \cref{thm:repapx_ML_existence,thm:repapx_SL_existence}, there exists an $\veps_1^2$-RepApx Maximal Lottery with support size $O(\veps_1^{-4})$ and an $\veps_2$-RepApx Stable $k$-Lottery with support size $O(k^2\veps_2^{-2})$. The following corollary therefore follows immediately.

\begin{corollary}\label{cor:mixing_thm} There exists a randomized voting rule with metric distortion strictly less than $3$ whose support size is at most a constant (independent of $n$ and $m$).
\end{corollary}

\begin{remark}
Both RepApx Maximal Lotteries and \nameref{box:repapx_pruned_lotteries} select a winner uniformly at random from a multiset. Since the mixing probability $\mu$ can be chosen to be rational, every candidate can be selected with rational probability. Consequently, the combined rule is equivalent to selecting uniformly at random from a list of constant size (with possible repetitions).
\end{remark}
 
\section{Concluding Thoughts}

In this work, we prove that a voting rule that uniformly randomizes over a (large) constant number of options achieves a metric distortion constant (slightly) less than $3$. While the number of options $N$ is an absolute constant, it remains quite large in our current proofs.

A natural question is to determine the smallest $N$ required to break the distortion barrier of $3$. It is known that deterministic rules (i.e., $N = 1$) are insufficient \cite{DBLP:journals/ai/AnshelevichBEPS18}. However, it seems implausible that the minimal necessary $N$ is large; in fact, we conjecture that $N = 2$ suffices to break the barrier of $3$---perhaps the Romans were right after all.

Very recently, Jannik Peters \cite{DBLP:journals/corr/peters2026} made interesting observations on the connection between metric distortion and undominated committees \cite{DBLP:journals/scw/ElkindLS15,DBLP:conf/stoc/CharikarLRV025}. His work shows that the support of any randomized voting rule with metric distortion better than $3- \veps$ must be $\left(1 - \frac{\veps}{2}\right)$-undominated. Notably, the existence of such a size-two committee was established only recently in the work of \cite{DBLP:conf/stoc/CharikarLRV025}. Together with the lower bound \cite{math/Zbarsky14} that $\alpha$-undominated committees of size $N$ are not guaranteed to exist when $\alpha < \frac{2}{N+1}$, Peters' work also implies a trade-off between the support size $N$ and the distortion parameter $\veps$. In particular, for $N = 2$, the metric distortion cannot be smaller than $3 - \frac23 \approx 2.33$.

Unfortunately, our current analytical framework is not remotely tight enough to establish such a small constant $N$, even for an improvement to $2.99$ in metric distortion; therefore, proving this conjecture will likely require significant new insights. Optimistically, one might hope that these insights are along the line of those needed to resolve the optimal metric distortion of randomized voting rules.

In a similar vein, Gkatzelis, Halpern, and Shah \cite{DBLP:conf/focs/GkatzelisHS20} provided the first proof that a deterministic voting rule can achieve distortion $3$, though their rule is complex and impractical in many scenarios. Subsequently, K{\i}z{\i}lkaya and Kempe \cite{DBLP:conf/ijcai/KizilkayaK22,DBLP:conf/sigecom/Kizilkaya023} introduced much simpler voting rules with the same distortion guarantee, which have elegant mathematical proofs and are more suitable for real-world applications. We hope our work stimulates a similar trajectory of simplification and refinement for randomized rules with restricted forms of randomness.

Finally, as noted in the introduction, our results have direct implications for the committee selection setting of \cite{DBLP:journals/corr/abs-2507-17063}. An interesting direction for future research is to explore the extent to which our techniques can be adapted to these multi-winner election settings.
 
\paragraph{Acknowledgments.} We thank Aris Filos-Ratsikas and Jannik Peters for their very helpful comments.

\bibliographystyle{alpha}
\newcommand{\etalchar}[1]{$^{#1}$}

\appendix

\section{Missing Details from Proof of \texorpdfstring{\cref{lem:sl_quasi_kernel}}{a Lemma}}
\label{sec:missing_sl_quasi_kernel}
\begin{claim}\label{clm:sl_quasi_kernel_details} If $\frac1{k + 1} + \veps \leq \frac2k$ and $k$ is an integer no less than $7$, then the minimum $\lambda$ that satisfies
\[
  p \leq \max\set{\frac{\lambda}{\theta}(1-\theta),\frac{\lambda}{\theta}\left(1-\left(\frac{1}{k+1}+\veps\right)p^{-k}\right)}, \qquad \forall p\in [0,1]
\]
is
\[
 \lambda(\theta, k, \veps)=\frac{\theta}{1-\theta}\left(\frac{1}{\theta (k+1)}+\frac{\veps}{\theta}\right)^{1/k}.
\]
\end{claim}
\begin{proof}[Proof of \cref{clm:sl_quasi_kernel_details}]

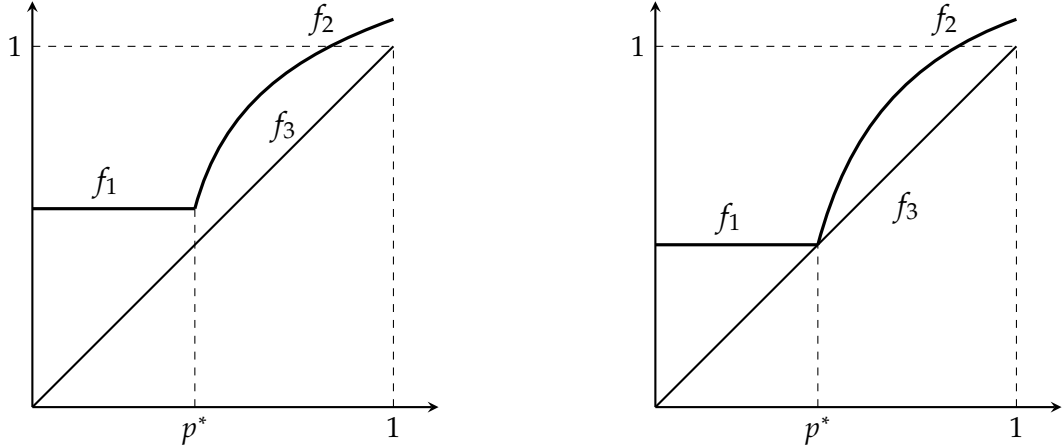
\begin{figure}[t!]
    \centering
    \begin{subfigure}[t]{0.49\textwidth}
        \centering
        \begin{tikzpicture}[scale=1.2, >=stealth]
\def\xmax{4}
    \def\ymax{4}
    \def\cornerX{1.8} \def\cornerY{2.2} 

\draw[->, thick] (0,0) -- (\xmax+0.5, 0) node[right] {};
    \draw[->, thick] (0,0) -- (0, \ymax+0.5) node[above] {};

\draw[dashed] (0, \ymax) -- (\xmax, \ymax);
\draw[dashed] (\xmax, 0) -- (\xmax, \xmax);
\draw[dashed] (\cornerX, 0) -- (\cornerX, \cornerY);
    
\node[left] at (0, \ymax) {$1$};
    \node[below] at (\xmax, 0) {$1$};
    \node[below] at (\cornerX, 0) {$p^*$};

\draw[thick] (0,0) -- (\xmax, \xmax);
    \node[below right, font=\large] at (2.5, 3.4) {$f_3$};

\draw[very thick] (0, \cornerY) -- (\cornerX, \cornerY);
    \node[above, font=\large] at (0.8, \cornerY) {$f_1$};

\draw[very thick] (\cornerX, \cornerY) to[out=75, in=200] (\xmax, \ymax+0.3);
    \node[above, font=\large] at (3.2, 4) {$f_2$};

\end{tikzpicture}         \caption{$f_1$ and $f_3$ do not intersect.}
        \label{fig:appendix_1}
    \end{subfigure}~ 
    \begin{subfigure}[t]{0.49\textwidth}
        \centering
        \begin{tikzpicture}[scale=1.2, >=stealth]
\def\xmax{4}
    \def\ymax{4}
    \def\cornerX{1.8} \def\cornerY{1.8} 

\draw[->, thick] (0,0) -- (\xmax+0.5, 0) node[right] {};
    \draw[->, thick] (0,0) -- (0, \ymax+0.5) node[above] {};

\draw[dashed] (0, \ymax) -- (\xmax, \ymax);
\draw[dashed] (\xmax, 0) -- (\xmax, \xmax);
\draw[dashed] (\cornerX, 0) -- (\cornerX, \cornerY);
    
\node[left] at (0, \ymax) {$1$};
    \node[below] at (\xmax, 0) {$1$};
    \node[below] at (\cornerX, 0) {$p^*$};

\draw[thick] (0,0) -- (\xmax, \xmax);
    \node[below right, font=\large] at (2.5, 2.5) {$f_3$};

\draw[very thick] (0, \cornerY) -- (\cornerX, \cornerY);
    \node[above, font=\large] at (0.8, \cornerY) {$f_1$};

\draw[very thick] (\cornerX, \cornerY) to[out=75, in=200] (\xmax, \ymax+0.3);
    \node[above, font=\large] at (3.2, 4) {$f_2$};

\end{tikzpicture}         \caption{$f_1$, $f_2$, and $f_3$ intersect at the same point.}
        \label{fig:appendix_2}
    \end{subfigure}
    \caption{Pictorial illustration of $f_1(p)$, $f_2(p)$, and $f_3(p)$.}
\end{figure}

For ease of notation, let $f_1(p)=\frac{\lambda}{\theta}(1-\theta)$, $f_2(p)=\frac{\lambda}{\theta}\cdot\left(1-\left(\frac{1}{k+1}+\veps\right)p^{-k}\right)$, and $f_3(p) = p$.

Our goal is to find the minimum $\lambda$ such that
\[
f_3(p) \leq \max\set{f_1(p), f_2(p)} \qquad \text{for all }p \in [0, 1].
\]

The intersection of $f_1(p)$ and $f_2(p)$ occurs at
\[
p^* = \left(\frac{1}{\theta (k+1)}+\frac{\veps}{\theta}\right)^{1/k},
\] which is in $(0, 1]$ under the assumption $\frac1{k + 1} + \veps \leq \frac2k$ and $k > 6$.

\cref{fig:appendix_1} illustrates the relevant geometry. Since
\[
\frac{\diff^2}{\diff p^2} \left(f_2 - f_3\right) = -\frac{\lambda k(k + 1)}{\theta}\left(\frac1{k + 1} + \veps\right)p^{-(k + 2)} < 0,
\] the function $f_2 - f_3$ is concave. Consequently, the condition $f_3(p) \leq \max\set{f_1(p), f_2(p)}$ holds for all $p \in [0, 1]$ if and only if
\[
f_1(p^*) \geq f_3(p^*)=p^* \qquad \text{and} \qquad f_2(1) \geq f_3(1)=1.
\] The second constraint is equivalent to
\[
\lambda \geq \frac{\theta}{1 - \left(\frac1{k + 1} + \veps\right)}.
\]

As $\lambda$ decreases, both $f_1$ and $f_2$ scale linearly downward. At the critical value $\lambda =\lambda^*$, all three functions intersect at $p^*$ (see \cref{fig:appendix_2}). For any $\lambda < \lambda ^*$, we have $f_1(p^*) < f_3(p^*)$, which violates the first constraint. Therefore, a valid $\lambda$ cannot fall below $\lambda ^*$.

It remains to verify that this $\lambda^*$ satisfies the second constraint. When $\lambda = \lambda^*$, we have
\[
\frac{\lambda}\theta(1 - \theta) = f_1(p^*) = f_3(p^*) = p^*, \text{and thus } \lambda = p^* \cdot \frac{\theta}{1 - \theta} = \left(\frac{\frac1{k + 1} + \veps}{\theta}\right)^{1/k} \cdot \frac{\theta}{1 - \theta}.
\] Therefore, the second constraint reduces to
\[
 \left(\frac{\frac1{k + 1} + \veps}{\theta}\right)^{1/k} \cdot \left(1 - \left(\frac1{k + 1} + \veps\right)\right) \geq 1 - \theta.
\]

It can be verified by computing
\begin{align*}
& \left(\frac{\frac1{k + 1} + \veps}{\theta}\right)^{1/k} \cdot \left(1 - \left(\frac1{k + 1} + \veps\right)\right)\\
\geq {}&{} \left(\frac1{k + 1} + \veps\right)^{1/k} \cdot \left(1 - \left(\frac1{k + 1} + \veps\right)\right) \tag{$\theta \leq 1$}\\
\geq {}&{} \left(\frac1{k + 1}\right)^{1/k} \cdot \left(1 - \left(\frac1{k + 1} + \veps\right)\right) \\
\geq {}&{} \left(\frac1{k + 1}\right)^{1/k} \cdot \left(1 - \frac2k\right) \tag{assumption that $\frac1{k + 1} + \veps \leq \frac2k$}\\
\geq {}&{} \frac12 \tag{can be confirmed by a standard derivative check if integer $k \geq 7$}\\
\geq {}&{} 1 - \theta \tag{$\theta > 1/2$}
\end{align*} This completes the proof.
\end{proof}

\section{Missing Details from Proof of \texorpdfstring{\cref{thm:mixing_thm}}{a Theorem}}
We restate the definition of $\lambda(\theta, k, \veps)$, which is valid for all integer $k \geq 7$:
\[
\lambda(\theta, k, \veps) = \frac{\theta}{1-\theta}\left(\frac{1}{\theta (k+1)}+\frac{\veps}{\theta}\right)^{1/k}.
\]

Throughout this section, we will use $L_k$ for $\ln(k / 4) / k$. Then
\[
\frac{\diff}{\diff k}L_k = \frac{1 - \ln(k/4)}{k^2},
\] so $L_k$ is maximized at $k = 4e$ when $k > 0$, with $L_{4e} = \frac1{4e}$.

\label{sec:missing_mixing_thm}
\begin{claim}\label{clm:lambda} Let $\tilde{\beta} = \frac19 L_k$ and $\veps_2 = \frac1k$. For any integer $k \geq 7$, we have
  \[
1+2\lambda\left(\frac12 + \tilde{\beta}, k, \veps_2\right) \leq 3 - L_k + L_k^2 \leq 3.
  \]
\end{claim}
\begin{proof}[Proof of \cref{clm:lambda}]
  We can compute
  \begin{align*}
  & 1+2\lambda\left(\frac12 + \tilde{\beta}, k, \veps_2\right) \\
  = {}&{} 1+ 2 \cdot \frac{1 + 2 \tilde{\beta}}{1 - 2\tilde{\beta}} \cdot \left(\frac1{\frac12 + \tilde{\beta}} \cdot \left(\frac1{k +1} + \frac1k\right)\right)^{1/k}\\
  \leq {}&{} 1 + 2 \cdot \frac{1 + 2 \tilde{\beta}}{1 - 2\tilde{\beta}} \cdot \left(2 \cdot \frac2k\right)^{1/k}\\
  =  {}&{} 1 + 2 \cdot \frac{1 + 2 \tilde{\beta}}{1 - 2\tilde{\beta}} \cdot \left(\frac4k\right)^{1/k}\\
  \leq {}&{} 1 + \left(2 + 9\tilde{\beta}\right)\cdot \left(\frac4k\right)^{1/k} \tag{$\frac{2(1+2x)}{1-2x} \leq 2+9x$ for any $x \in [0, \frac1{18}]$ and $0 \leq \tilde{\beta} = \frac19L_k \leq \frac1{36e} \leq \frac1{18}$ for any integer $k \geq 4$}\\
  = {}&{} 1 + (2 + L_k) \cdot \exp(-L_k)\\
  \leq {}&{} 1 + (2 + L_k) \cdot \left(1 - L_k + \frac12L_k^2\right) \tag{$\exp(-x) \leq 1 - x +\frac12x^2$ for any $x \geq 0$ and $L_k \geq 0$ for any integer $k \geq 4$}\\
  \leq {}&{} 3 -2L_k + L_k^2 + L_k - L_k \cdot \left(L_k - \frac12 L_k^2\right)\\
  = {}&{} 3 - L_k + L_k^2 - L_k \cdot \left(L_k - \frac12 L_k^2\right)\\
  \leq {}&{} 3 - L_k + L_k^2 \tag{$x - \frac12x^2 \geq 0$ for any $x \in [0, 2]$ and $0 \leq L_k \leq \frac1{4e} \leq 2$ for any integer $k \geq 4$}\\
  \leq {}&{} 3. \tag{$x - x^2 \geq 0$ for any $x \in [0, 1]$ and $0 \leq L_k \leq \frac1{4e} \leq 1$ for any integer $k \geq 4$}
  \end{align*} This completes the proof.
\end{proof}

The following two results are corollaries of \cref{clm:lambda}. The proof of \cref{clm:1} is immediate. We will only prove \cref{clm:2}.

\begin{claim}\label{clm:1}
Let $\tilde{\beta} = \frac19 L_k$, $\veps_2 = \frac1k$. For any integer $k \geq 7$, we have
  \[
\left(1+2\lambda\left(\frac12 + \tilde{\beta}, k, \veps_2\right)\right)\left(\frac4{\frac12 + \tilde{\beta}}-3\right) \leq 15.
  \]
\end{claim}

\begin{claim}\label{clm:2}
Let $\alpha = \frac1{24}L_k$, $\tilde{\beta} = \frac19 L_k$, and $\veps_2 = \frac1k$. For any integer $k \geq 7$, we have
  \[
\left(1+2\lambda\left(\frac12 + \tilde{\beta}, k, \veps_2\right)\right)(1 + 4\alpha) \leq 3 - \frac12 L_k + L_k^2.
  \]
\end{claim}
\begin{proof}[Proof of \cref{clm:2}]
  We can compute
  \begin{align*}
  & \left(1+2\lambda\left(\frac12 + \tilde{\beta}, k, \veps_2\right)\right)(1 + 4\alpha)\\
  \leq {}&{} \left(3 - L_k + L_k^2\right)(1 + 4\alpha)\tag{the first inequality in \cref{clm:lambda}}\\
  = {}&{} 3 - L_k + L_k^2 + 4\alpha \cdot (3 - L_k + L_k^2) \\
  \leq {}&{} 3 - L_k + L_k^2 + 12\alpha \tag{the second inequality in \cref{clm:lambda}}\\
  = {}&{} 3 - L_k + L_k^2 + \frac12 L_k\\
  = {}&{} 3 - \frac12L_k + L_k^2 
  \end{align*} This completes the proof.
\end{proof}

We now show the sufficient condition when the biased metric is strongly $\left(\alpha, \frac12 + \tilde{\beta}\right)$-consistent.

\begin{claim}\label{clm:cond1}
  Let $\alpha = \frac1{24}L_k$, $\tilde{\beta} = \frac19 L_k$, $\veps_1 = \frac1{150000}L_k^3$, $\veps_2 = \frac1k$, and $\mu = 1 - \frac1{2000}L_k^2$. For any integer $k \geq 7$, we have
  \[
  \mu (3+28\veps_1) + (1 - \mu) \left(1+2\lambda\left(\frac12 + \tilde{\beta}, k, \veps_2\right)\right)(1+4\alpha) < 3.
  \]
\end{claim}
\begin{proof}[Proof of \cref{clm:cond1}]
We can compute
\begin{align*}
& \mu (3+28\veps_1) + (1 - \mu) \left(1+2\lambda\left(\frac12 + \tilde{\beta}, k, \veps_2\right)\right)(1+4\alpha)\\
\leq {}&{} \mu (3 + 28\veps_1) + (1 - \mu)\left(3 - \frac12L_k + L_k^2 \right) \tag{\cref{clm:2}}\\
= {}&{} 3 + 28\mu\veps_1 + (1 - \mu) \left( - \frac12L_k + L_k^2 \right)\\
\leq {}&{} 3 + 28\veps_1 + (1 - \mu) \left( - \frac12L_k + L_k^2 \right)\\
= {}&{} 3 + \frac{28}{150000}L_k^3 +\frac1{2000}L_k^2 \cdot \left( - \frac12L_k + L_k^2 \right)\\
= {}&{} 3 + L_k^3 \left(\frac{28}{150000} + \frac{L_k - \frac12}{2000}\right)\\
\leq {}&{} 3 + L_k^3 \left(\frac{28}{150000} + \frac{\frac1{4e} - \frac12}{2000}\right) \tag{$L_k \leq \frac1{4e}$ for any integer $k > 0$}\\
< {}&{} 3 \tag{$L_k > 0$ for any integer $k \geq 5$}.
\end{align*} This completes the proof.
\end{proof}

Next, we show the sufficient condition when the biased metric is not strongly $\left(\alpha, \frac12 + \tilde{\beta}\right)$-consistent.

\begin{claim}\label{clm:cond2}
 Let $\alpha = \frac1{24}L_k$, $\tilde{\beta} = \frac19 L_k$, $\veps_1 = \frac1{150000}L_k^3$, $\veps_2 = \frac1k$, and $\mu = 1 - \frac1{2000}L_k^2$. For any integer $k \geq 7$, we have
  \[
  \mu(3 + 28\veps_1 - 2\alpha\tilde{\beta}) + (1 - \mu)\left(1+2\lambda\left(\frac12 + \tilde{\beta}, k, \veps_2\right)\right)\left(\frac{4}{1/2 + \tilde{\beta}} - 3\right) < 3.
  \]
\end{claim}
\begin{proof}[Proof of \cref{clm:cond2}]
  We can compute
  \begin{align*}
 & \mu(3 + 28\veps_1 - 2\alpha\tilde{\beta}) + (1 - \mu)\left(1+2\lambda\left(\frac12 + \tilde{\beta}, k, \veps_2\right)\right)\left(\frac{4}{1/2 + \tilde{\beta}} - 3\right)\\
 \leq {}&{} \mu(3 + 28\veps_1 - 2\alpha\tilde{\beta}) + 15(1 - \mu) \tag{\cref{clm:1}}\\
 = {}&{} 3 + 12 (1 - \mu) - \mu(2\alpha\tilde{\beta} - 28\veps_1)\\
 = {}&{} 3 + 12 (1 - \mu) - \mu\left(\frac1{108}L_k^2 - 28\veps_1\right)\\
 = {}&{}  3 + \frac{12}{2000}L_k^2 - \mu\left(\frac1{108}L_k^2 - \frac{28}{150000}L_k^3\right)\\
 ={}&{} 3 + L_k^2\left(\frac{12}{2000} - \mu\left(\frac1{108} - \frac{28}{150000} L_k\right)\right)\\
 \leq {}&{} 3 + L_k^2\left(\frac{12}{2000} - \mu\left(\frac1{108} - \frac{28}{150000 \cdot 4e}\right)\right) \tag{$L_k \leq \frac1{4e}$ for any integer $k > 0$}\\
 \leq {}&{} 3 +L_k^2\left(\frac{12}{2000} - \left(1 - \frac1{2000}\cdot \left(\frac1{4e}\right)^2\right) \cdot \left(\frac1{108} - \frac{28}{150000 \cdot 4e}\right)\right) \tag{$\mu = 1 - \frac1{2000}L_k^2 \geq 1 - \frac1{2000}\cdot \left(\frac1{4e}\right)^2$ for any integer $k \geq 4$}\\
 < {}&{} 3\tag{$L_k > 0$ for any integer $k \geq 5$}.
  \end{align*} This completes the proof.
\end{proof}
 
\end{document}